\documentclass[12pt,a4paper]{article}

\usepackage[utf8]{inputenc} 
\usepackage[T1]{fontenc}    
\usepackage[pagebackref]{hyperref}       
\usepackage{url}            
\usepackage{booktabs}       
\usepackage{amsfonts}       
\usepackage{nicefrac}       
\usepackage{microtype}      
\usepackage{xcolor}         
\usepackage{authblk}
\usepackage{xspace}
\usepackage{fullpage}
\usepackage{lmodern}
\usepackage[noend,ruled,algo2e]{algorithm2e}
\usepackage{dsfont}

\usepackage{amsmath, amssymb, amsthm}
\usepackage{mathtools}
\usepackage{tikz}
\usepackage[capitalize, nameinlink]{cleveref}

\usepackage{todonotes}
\usepackage{paralist}
\usepackage{booktabs}
\usepackage{caption}
\usepackage[numbers]{natbib}

\crefformat{footnote}{#2\footnotemark[#1]#3}

\newtheorem{theorem}{Theorem}

\newtheorem{corollary}[theorem]{Corollary}

\newtheorem{lemma}[theorem]{Lemma}
\newtheorem{observation}[theorem]{Observation}
\newtheorem{definition}[theorem]{Definition}
\Crefname{theorem}{Theorem}{Theorems}
\Crefname{claim}{Claim}{Claims}
\Crefname{proposition}{Proposition}{Propositions}
\Crefname{corollary}{Corollary}{Corollaries}
\Crefname{algorithm}{Algorithm}{Algorithms}
\Crefname{lemma}{Lemma}{Lemmas}
\Crefname{definition}{Definition}{Definitions}
\Crefname{observation}{Observation}{Observations}

\DeclareMathOperator{\poly}{poly}
\DeclareMathOperator{\sgn}{sgn}
\DeclareMathOperator{\aff}{aff}
\newcommand{\N}{\mathbb{N}}

\newcommand{\R}{\mathbb{R}}

\newcommand{\RELU}[1]{\textsc{2L-ReLU-NN-Train(\ensuremath{#1})}}
\newcommand{\LT}[1]{\textsc{2L-LT-NN-Train(\ensuremath{#1})}}

\newcommand{\problemdef}[3]{
	\begin{center}
		\begin{minipage}{0.95\textwidth}
			\noindent
			\textsc{#1}
			
			\vspace{2pt}
			\setlength{\tabcolsep}{5pt}
			\begin{tabular}{p{0.13\textwidth}p{0.82\textwidth}}
				\textbf{Input:} 		& #2 \\
				\textbf{Question:} 	& #3
			\end{tabular}
		\end{minipage}
	\end{center}
}

\usetikzlibrary{arrows.meta, shapes}
\tikzset{bigneuron/.style={circle, draw, inner sep = 0, minimum width = 5.5ex}}
\tikzset{smallneuron/.style={circle, draw, inner sep = 0, minimum width = 4ex}}
\tikzset{transform/.style={fill=white, circle}}
\tikzset{connection/.style={-{Stealth}}}
\newcommand{\relu}{
	\begin{tikzpicture}
		\draw [line width=1pt] (-1.1ex,0) -- (0,0) -- (0.9ex,0.9ex);
	\end{tikzpicture}
}

\title{Training Neural Networks is NP-Hard in Fixed Dimension}

\author[1]{Vincent Froese}
\author[2]{Christoph Hertrich}

\affil[1]{\small
  Algorithmics and Computational Complexity,\protect\\ Faculty~IV, TU Berlin, Berlin, Germany,\protect\\
  vincent.froese@tu-berlin.de}

\affil[2]{Department of Mathematics,\protect\\ London School of Economics and Political Science, London, UK,\protect\\c.hertrich@lse.ac.uk}

\begin{document}

\maketitle

\begin{abstract}
  We study the parameterized complexity of training two-layer neural networks with respect to the dimension of the input data and the number of hidden neurons, considering ReLU and linear threshold activation functions.
  Albeit the computational complexity of these problems has been studied numerous times in recent years, several questions are still open.
  We answer questions by Arora et al.~[ICLR~'18] and Khalife and Basu~[IPCO~'22] showing that both problems are NP-hard for two dimensions, which excludes any polynomial-time algorithm for constant dimension.
  We also answer a question by Froese et al.~[JAIR~'22] proving W[1]-hardness for four ReLUs (or two linear threshold neurons) with zero training error.
  Finally, in the ReLU case, we show fixed-parameter tractability for the combined parameter number of dimensions and number of ReLUs if the network is assumed to compute a convex map.
  Our results settle the complexity status regarding these parameters almost completely.
\end{abstract}

\section{Introduction}

Neural networks with rectified linear unit (ReLU) activations are arguably one of the most fundamental models in modern machine learning~\cite{ABMM18,GBB11,LBH25}. To use them as predictors on unseen data, one usually first fixes an architecture (the graph of the neural network) and then optimizes the weights and biases such that the network performs well on some known training data, with the hope that it will then also generalize well to unseen test data. While the ultimate goal in applications is generalization, \emph{empirical risk minimization} (that is, optimizing the training error) is an important step in this pipeline and understanding its computational complexity is crucial to advance the theoretical foundations of deep learning.

In this paper, we aim to understand how the choice of different meta-parameters, like the input dimension and the width of the neural network, influences the computational complexity of the training problem. To this end, we focus on two-layer neural networks, which can be seen as the standard building block also for deeper architectures.

Formally, a two-layer neural network (see \Cref{Fig:NNArchitecture}) with~$d$ input neurons, $k$ hidden ReLU neurons, and a single output neuron computes a map
\[\phi\colon\R^d\to\R, \quad \phi(\mathbf{x})=\sum_{j=1}^k a_j[\mathbf{w}_j\cdot\mathbf{x} + b_j]_+,\]
where $\mathbf{w}_j\in\R^d$ and~$a_j\in\{-1,1\}$ are the weights between the layers, $b_j$ are the biases at the hidden neurons, and~$[x]_+ \coloneqq \max(0,x)$ is the \emph{rectifier} function. Notice that restricting~$a_j$ to~$\{-1,1\}$ is without loss of generality because we can normalize by pulling any nonnegative factor into $\mathbf w_j$ and $b_j$.
We also study neural networks with \emph{linear threshold activation} in \Cref{sec:LT}.

Given training data $\mathbf{x}_1,\ldots,\mathbf{x}_n\in\R^d$ with labels~$y_1,\ldots,y_n\in\R$, the task of training such a network is to find~$\mathbf{w}_j,b_j,$ and $a_j$ for each~$j\in[k]$ such that the training error~$\sum_{i=1}^n \mathcal{L}(\phi(\mathbf{x}_i),\ y_i)$
for a given loss function $\mathcal{L}\colon\R\times\R\to\R_{\geq0}$ is minimized.
Formally, the decision version of two-layer ReLU neural network training is defined as follows:

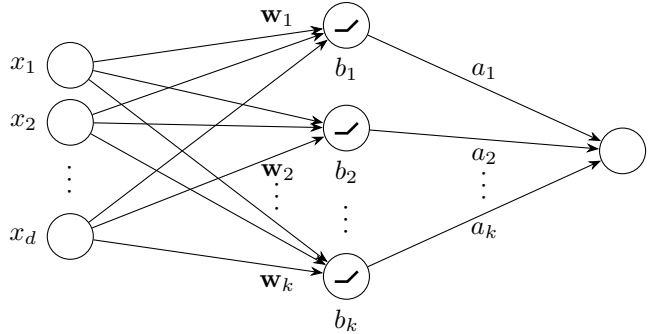
\begin{figure}[t]
	\centering
	\begin{tikzpicture}[]
		\footnotesize
		\node[smallneuron, label=below:{$b_1$}] (n1) at (0,11ex) {\relu};
		\node[smallneuron, label=below:{$b_2$}] (n2) at (0,2ex) {\relu};
		\node[rotate=90] (npunkt) at (0,-6ex) {$\mathbf{\cdots}$};
		\node[smallneuron, label=below:{$b_k$}] (n3) at (0,-11ex) {\relu};
		\node[smallneuron, label=left:{$x_1$}] (in1) at (-24ex,7.5ex) {};
		\node[smallneuron, label=left:{$x_2$}] (in2) at (-24ex,2.5ex) {};
		\node[rotate=90] (inpunkt) at (-24ex,-2.5ex) {$\mathbf{\cdots}$};
		\node[smallneuron, label=left:{$x_d$}] (in3) at (-24ex,-7.5ex) {};
		\node[smallneuron] (out) at (24ex,0) {};
		\draw[connection] (in1) -- (n1) node[above, pos=0.8] {$\mathbf w_1$};
		\draw[connection] (in2) -- (n1);
		\draw[connection] (in3) -- (n1);
		\draw[connection] (in1) -- (n2);
		\draw[connection] (in2) -- (n2);
		\draw[connection] (in3) -- (n2) node[below, pos=0.8] {$\mathbf w_2$};
		\draw[connection] (in1) -- (n3);
		\draw[connection] (in2) -- (n3);
		\draw[connection] (in3) -- (n3) node[below, pos=0.8] {$\mathbf w_k$};
                \node[rotate=90] (npunkt) at (-6ex,-4ex) {$\mathbf{\cdots}$};
		\draw[connection] (n1) -- (out) node[above, pos=0.5] {$a_1$};
		\draw[connection] (n2) -- (out) node[below, pos=0.5] {$a_2$};
                \node[rotate=90] (npunkt) at (12ex,-3ex) {$\mathbf{\cdots}$};
		\draw[connection] (n3) -- (out) node[below, pos=0.5] {$a_k$};
	\end{tikzpicture}
	\caption{Neural network architecture we study in this paper: After the input layer (left) with $d$ input neurons, we have one hidden layer with $k$ ReLU neurons and a single output neuron without additional activation function.}
	\label{Fig:NNArchitecture}
\end{figure}

\problemdef{\RELU{\mathcal{L}}}
{Data points~$(\mathbf{x}_1,y_1),\ldots,(\mathbf{x}_n,y_n)\in \R^d\times\R$, a number~$k\in\N$ of ReLUs, and a target error~$\gamma\in\R_{\ge 0}$.}
{Are there weights $\mathbf w_1, \ldots, \mathbf w_k\in \R^d$, biases $b_1,\ldots,b_k\in\R$, and coefficients $a_1,\ldots,a_k\in\{-1,1\}$ such that
	\[
	\sum_{i=1}^n \mathcal{L}\Bigg(\sum_{j=1}^k a_j[\mathbf{w}_j\cdot\mathbf{x}_i + b_j]_+,\ y_i\Bigg) \le \gamma?
	\]
      }
      Note that in the over-parameterized case where~$k \ge n$, the network can exactly fit any~$n$ input points\footnote{Assuming that~$y_i=y_j$ whenever $\mathbf{x}_i=\mathbf{x}_j$.} (achieving training error~$\gamma=0$)~\cite[Theorem~1]{ZBHRV21}.
      Thus, we henceforth assume that~$k <n$.
  
  \RELU{\mathcal{L}} is known to be NP-hard~\cite{DWX20,GKMR21}, but all known reductions require the input dimension to be part of the input. The current state-of-the-art exact algorithm for convex loss~$\mathcal{L}$ is by \citet{ABMM18} and runs in $O(2^kn^{dk}\poly(L))$ time, where~$L$ is the input bit-length.

As regards the computational complexity of \RELU{\mathcal{L}}, \citet{ABMM18} posed the question(s) whether a running time
\begin{quote}
  \emph{``that is polynomial in the data size and/or the number of hidden nodes, assuming that the input dimension is a fixed constant''}
\end{quote}
is possible.
That is, they asked two questions. The first corresponds to the ``and'' statement, which can be phrased as follows:\\[1em]
\textbf{Question 1}: Is there an algorithm running in~$(nk)^{f(d)}\poly(L)$ time for some function~$f$?\\[1em]
\noindent
In other words, the question is whether \RELU{\mathcal{L}} is in the complexity class XP when parameterized by~$d$.
The second question corresponding to the ``or'' statement can then be interpreted as\\[1em]
\textbf{Question 2}: Is there an algorithm running in~$n^{f(d)}g(k,d)\poly(L)$ or~$k^{f(d)}g(n,d)\poly(L)$ time for some functions~$f$ and~$g$?\\[1em]
\noindent
We observe that the second running time is clearly possible since $k<n$ holds by assumption, and hence the algorithm by~\citet{ABMM18} runs in $g(n,d)\poly(L)$ time.
Hence, it remains open whether $n^{f(d)}g(k,d)\poly(L)$ time is possible, which is equivalent to (uniform) fixed-parameter tractability with respect to~$k$ for every constant~$d$.

Clearly, Question~1 is the stronger statement, that is, a positive answer implies a positive answer to Question~2.
\citet{ABMM18} conclude with
\begin{quote}
  \emph{``Resolving this dependence on network size would be another step towards clarifying the theoretical complexity of training ReLU DNNs and is a good open question for future research, in our opinion.''}
\end{quote}

Note that \citet{FHN22} proved that, for~$k=1$, there is no algorithm running in~$g(d)n^{o(d)}$ time unless the Exponential Time Hypothesis fails.
Hence, this result already partially answered the two questions above by excluding any algorithm running in $n^{o(d)}g(d,k)\poly(L)$ time.

In this paper, we answer Question~1 negatively by showing NP-hardness for~$d=2$ in~\Cref{thm:d=2}, indicating that we cannot get rid of the exponential dependence on the network size in the algorithm by Arora et al.~\cite{ABMM18} even if the dimension is fixed.
As regards Question~2, we further exclude (assuming the Exponential Time Hypothesis) any algorithm running in time $n^{o(d)}g(d,k)\poly(L)$ even for the case~$\gamma=0$ and prove W[1]-hardness with respect to~$d$ for~$k=4$ (\Cref{thm:k=4}), which answers an open question by Froese et al.~\cite{FHN22}.

We also obtain analogous hardness results if linear threshold activation functions are used instead of ReLUs. As in the ReLU case, it is well-known that training linear threshold networks is NP-hard~\cite{BR92,KB22}. The running time of the state-of-the-art algorithm due to \citet{KB22} is polynomial in $n$ for fixed $d$ and $k$, but exponential in the latter two parameters. \citet{KB22} posed an analogous question to Question~1 for linear threshold networks, which we answer negatively in \Cref{cor:d=2}, excluding a polynomial running time even for fixed dimension. We also show that we cannot expect fixed-parameter tractability with respect to~$d$ even for $k=1$ (\Cref{cor:k=1}) and also not for $k=2$ and $\gamma=0$ (\Cref{cor:k=2}).

On the positive side, we give an algorithm running in~$2^{O(k^2d)}\poly(k,L)$ time for ReLU neural networks if~$\gamma=0$ and the function computed by the network is assumed to be convex (\Cref{thm:convexFPT}).
Note that this running time yields fixed-parameter tractability with respect to~$k$ for every constant~$d$, and thus answers Question~2 positively for this restricted special case.

\paragraph*{Implications and Limitations.}

In the following we provide a brief discussion of the implications and limitations of our results from various perspectives.

\subparagraph{Input Dimension.} \Cref{thm:d=2} implies that \RELU{\mathcal{L}} is in fact NP-hard for every fixed $d\geq 2$. The straight-forward reduction simply pads all the input vectors with $d-2$ zeros. Similarly, \Cref{cor:d=2} holds for every fixed $d\geq2$.
\subparagraph{Target Error.} The hardness results \Cref{thm:d=2,thm:k=4,cor:d=2,cor:k=2} also hold for every fixed $\gamma\geq 0$. The reduction is straight-forward by introducing a set of incompatible data points which force the network to incur an additional error of $\gamma$.  For our positive result \Cref{thm:convexFPT}, however, there is indeed a difference in the complexity between the two cases~$\gamma=0$ and~$\gamma>0$. While we show fixed-parameter tractability for $\gamma=0$, the same problem is W[1]-hard for $\gamma>0$, already in the case $k=1$~\cite{FHN22}.
\subparagraph{Number of ReLUs.} It is not too difficult to see (although it requires some work) that our particular reduction in \Cref{thm:k=4} can be extended to any $k\geq 4$ by introducing more data points far away from the existing data points which enforce the usage of additional ReLUs which then cannot be used to fit the data points of the actual reduction. Therefore, \Cref{thm:k=4} holds for every fixed $k\geq 4$. Similarly, \Cref{cor:k=1} holds for every fixed $k\geq1$ and \Cref{cor:k=2} holds for every fixed $k\geq2$.
\subparagraph{Other Activation Functions.} Our hardness results hold for the piecewise linear ReLU activation function and the piecewise constant linear threshold activation function. Extending them to other piecewise linear or constant activation functions like leaky ReLU or maxout should be straight-forward. However, achieving analogous results for smooth activation functions like sigmoids probably requires fundamentally different techniques and is beyond the scope of this paper.
\subparagraph{Training vs.\ Learning.} Our results are concerned with the problem of minimizing the training error. While this is inherently different from minimizing the generalization error, there are indeed deep connections between these two problems~\cite{SB14}. In particular, as pointed out by~\citet{GKMR21}, hardness of training implies hardness of proper learning if one permits arbitrary data distributions. However, such hardness results can often be overcome by either posing additional assumptions on the data distributions or switching to more general learning paradigms like improper learning~\cite{GKKT17}.
\subparagraph{Exact vs.\ Approximate Training.} In practice, it arguably often suffices to train a neural network to approximate instead of exact optimality. The results in this paper are concerned with solving the training problem to exact global optimality. However, since \Cref{thm:d=2,thm:k=4,cor:d=2,cor:k=2} already hold for training error $\gamma=0$, they even rule out the existence of approximation algorithms with any multiplicative factor. We conceive that for appropriate notions of \emph{additive} approximation (see, e.g., \cite{GKMR21}), our reductions can also be used to show hardness of additive approximation. However, this would significantly increase the technical complexity of the analysis and is therefore beyond the scope of this paper. We leave it as an open research question to analyze the influence of meta-parameters like input dimension and number of hidden neurons on additive approximation of the training problem.

\paragraph*{Related Work.}
\citet{DWX20} and \citet{GKMR21} showed NP-hardness of \RELU{\mathcal{L}} for~$k=1$, but require non-constant dimension.
For target error $\gamma=0$, the problem is NP-hard for every constant~$k\ge 2$ and polynomial-time solvable for~$k=1$~\cite{GKMR21}.
\citet{GKMR21} provide further conditional running time lower bounds and (additive) approximation hardness results.
\citet{FHN22} considered the parameterized complexity regarding the input dimension~$d$ and
proved W[1]-hardness and an ETH-based running time lower bound of~$n^{\Omega(d)}$ for~$k=1$.

\citet{BDL22} studied networks where the output neuron also is a ReLU and proved NP-hardness (and implicitly W[1]-hardness with respect to~$d$) for~$k=2$ and~$\gamma=0$.
\citet{BHJMW22} showed that training 2-layer ReLU networks with two output and two input neurons ($\R^2\to\R^2$) is complete for the class $\exists\R$ (existential theory of the
reals) and thus likely not contained in~NP. This also implies NP-hardness, but note that in contrast to our results, their reduction does not work for one-dimensional outputs. Their result strengthens a previous result by~\citet{AKM21} who proved $\exists\R$-completeness for networks with a specific (not fully connected) architecture.
\citet{PE20} showed that training 2-layer neural networks can be formulated as a convex program which yields a polynomial-time algorithm for constant dimension~$d$. However, they considered a regularized objective and their result requires the number~$k$ of hidden neurons to be very large (possibly equal to the number~$n$ of input points) and hence does not contradict our NP-hardness result for~$d=2$.

To study the computational complexity of training ReLU networks, a crucial ingredient is to know the set of (continuous and piecewise linear) functions precisely representable with a certain network architecture. This is well-understood for two-layer networks~\cite{ABMM18,BHJMW22,DK22}, but much trickier for deeper networks~\cite{HHL23,HS21,HBDS21}. Similar to the study of ReLU networks by Arora et al.~\cite{ABMM18}, Khalife and Basu~\cite{KB22} studied the expressiveness and training complexity for linear threshold activation functions.
For an extensive survey on intersections of deep learning and polyhedral theory, we refer to \citet{HMST23}.

\section{Preliminaries}\label{sec:prelims}

\paragraph*{Notation.} For $n\in\N$, we define~$[n]\coloneqq\{1,\ldots,n\}$.
For~$X\subseteq \R^d$, we denote by~$\aff(X)$ the affine hull of~$X$ and by~$\dim(X)$ the dimension of~$\aff(X)$.

Throughout this work, we assume~$\mathcal{L}\colon \R\times\R\to\R_{\geq0}$ to be any loss function with~$\mathcal{L}(x,y)=0 \iff x=y$.

\paragraph*{Parameterized Complexity.}
We assume basic knowledge on computational complexity theory.
Parameterized complexity is a multivariate approach to analyze the computational complexity of problems~\cite{DF13,Cyg15}.

An instance~$(x,k)$ of a parameterized problem~$L\subseteq\Sigma^*\times\N$ is a pair with~$x\in\Sigma^*$ being a problem instance and $k\in\N$ being the value of a certain~\emph{parameter}.
A parameterized problem~$L$ is \emph{fixed-parameter tractable (fpt)} (contained in the class FPT) if there exists an algorithm deciding whether~$(x,k)\in L$ in $f(k)\cdot|x|^{O(1)}$ time, where~$f$ is a function solely depending on~$k$.
Note that a parameterized problem in FPT is polynomial-time solvable for every constant parameter value where, importantly, the degree of the polynomial does not depend on the parameter value.
The class XP contains all parameterized problems which can be solved in polynomial time for constant parameter values, that is, in time~$f(k)\cdot |x|^{g(k)}$.
It is known that $\text{FPT}\subsetneq\text{XP}$.
The class W[1] contains parameterized problems which are widely believed not to be in FPT.
That is, a W[1]-hard problem (e.g.~\textsc{Clique} parameterized by the size of the sought clique)
is not solvable in $f(k)\cdot |x|^{O(1)}$ time.
It is known that $\text{FPT}\subseteq \text{W[1]}\subseteq \text{XP}$.

W[1]-hardness is defined via \emph{parameterized reductions}.
A parameterized reduction from~$L$ to~$L'$ is an algorithm mapping an instance~$(x,k)$ in $f(k)\cdot |x|^{O(1)}$ time to an instance~$(x',k')$ such that~$k'\le g(k)$ for some function~$g$ and $(x,k)\in L$ if and only if~$(x',k')\in L'$.

\paragraph*{Exponential Time Hypothesis.}
The Exponential Time Hypothesis (ETH)~\cite{IP01} states that \textsc{3-SAT} cannot be solved in subexponential time in the number~$n$ of Boolean variables in the input formula, that is, there exists a constant~$c > 0$ such that \textsc{3-SAT} cannot be solved in~$O(2^{cn})$ time.

The ETH implies $\text{FPT}\neq\text{W[1]}$~\cite{Cyg15} (which implies $\text{P}\neq\text{NP}$).
In fact, ETH implies that \textsc{Clique} cannot be solved in~$\rho(k)\cdot n^{o(k)}$ time for any function~$\rho$, where~$k$ is the size of the sought clique and $n$~is the number of vertices in the graph~\cite{Cyg15}.

\paragraph*{Geometry of 2-Layer ReLU Networks.}
For proving our results, it is crucial to understand the geometry of a function $\phi\colon\R^d\to\R$ represented by a two-layer ReLU network. Here, we only discuss properties required to understand our results and refer to~\cite{ABMM18,BHJMW22,DK22} for additional discussions in this context. Such a function $\phi$ is a continuous and piecewise linear function. Each hidden neuron with index $j\in[k]$ defines a hyperplane $\mathbf{w}_j\cdot\mathbf{x} + b_j=0$ in~$\R^d$. These~$k$ hyperplanes form a hyperplane arrangement. Inside each cell of this hyperplane arrangement, the function $\phi$ is affine. The graph of $\phi$ restricted to such a cell is called a \emph{(linear) piece} of $\phi$.

Consider a hyperplane~$H$ from the hyperplane arrangement. Let $\mathbf w$ be an orthonormal vector of $H$ and let $J\subseteq[k]$ be the non-empty subset of indices of neurons which induce precisely $H$. Note that $\mathbf{w}_j$ is a scaled version of $\mathbf w$ for each $j\in J$. Let $\mathbf x\in\R^d$ be a point on $H$ which does not lie on any other hyperplane in the arrangement. There are exactly two $d$-dimensional cells in the arrangement containing $\mathbf x$: one on each side of~$H$. The difference of the directional derivatives of the corresponding two pieces of $\phi$ in the direction of $\mathbf w$ is exactly $\sum_{j\in J} a_j\lVert\mathbf{w}_j\rVert$. In particular, this is independent of $\mathbf{x}$ and therefore constant along $H$. If this value is positive, we say that $H$ is a \emph{convex} hyperplane of $\phi$. If it is negative, we say that~$H$ is a \emph{concave} hyperplane of $\phi$. Note that this matches with~$\phi$ being convex or concave locally around every point $\mathbf x \in H$ which does not belong to any other hyperplane in the arrangement. Moreover, a point $\mathbf x\in\R^d$ is called a convex (concave) \emph{breakpoint} of $\phi$ if it lies exclusively on one convex (concave) hyperplane of $\phi$.

One important observation we will heavily use is the following: If we know that $\phi$ originates from a 2-layer neural network with $k$ hidden neurons and we know that we need indeed $k$ distinct hyperplanes to separate the pieces of $\phi$, then each hyperplane must be induced by exactly one neuron (and not several). Then the hyperplane corresponding to the $j$-th neuron is convex if and only if $a_j>0$ and concave if and only if $a_j<0$.
For input dimension $d=2$, each of the hyperplanes in the arrangement is actually a line in~$\R^2$. We call such a line a \emph{breakline} and define \emph{convex} and \emph{concave} breaklines accordingly.

\section{NP-Hardness for Two Dimensions}

In this section we prove our main result that \RELU{\mathcal{L}} is NP-hard for two dimensions, thus excluding any running time of the form~$(nk)^{f(d)}$.

\begin{theorem}\label{thm:d=2}
  \RELU{\mathcal{L}} is NP-hard even for~$d=2$ and~$\gamma=0$.
\end{theorem}

We give a polynomial-time reduction from the following NP-complete problem~\cite{GJ79}.

\problemdef{Positive One-In-Three 3-SAT (POITS)}
{A Boolean formula~$F$ in conjunctive normal form with three positive literals per clause.}
{Is there a truth assignment for the variables such that each clause in~$F$ has exactly one true literal?}

Our construction will be such that the function represented by the neural network is equal to zero everywhere except for a finite set of ``stripes'', in which the function forms a \emph{levee} (see \Cref{def:levee}), that is, when looking at a cross section, the function goes up from~0 to 1, stays constant for a while, and goes down from 1 to 0 again. See \Cref{fig:selectiongadget} (right) for a top view of a levee and \Cref{fig:levee_crosssection} for a cross section of a levee.

\begin{definition}\label{def:levee}
A \emph{levee} with slope $s\in\R$ (centered at the origin) is the function~$f_s\colon\R^2\to\R$ with
\begin{equation}\label{eq:levee}
f_s(x_1,x_2)=\left\{
\begin{array}{ll}
	0,&\text{if }|x_2-sx_1|\geq 2,\\
	1,&\text{if }|x_2-sx_1|\leq 1,\\
	2+x_2-sx_1,& \text{if }x_2-sx_1 \in \;]-2,-1[,\\
	2-x_2+sx_1,& \text{if }x_2-sx_1\in \;]1,2[.
\end{array}
\right.
\end{equation}
\end{definition}

\begin{observation}\label{obs:levee}
	A levee $f_s$ is a continuous, piecewise linear function with two convex and two concave breaklines. It can be realized with four ReLUs as follows:
	\[
		f_s(x_1,x_2) = [x_2-sx_1+2]_+ - [x_2-sx_1+1]_+ - [x_2-sx_1-1]_+ + [x_2-sx_1-2]_+.
	\]
\end{observation}

Similar levees have been used by \citet{BHJMW22} to prove $\exists\R$-completeness of neural network training, however, in a conceptually very different way. In their work, levees encode variable values via the \emph{slopes of the function} on the non-constant regions of the levee. In contrast, in our reduction, we encode discrete choices via rotation of the levees, that is, via the \emph{slopes of the breaklines in the two-dimensional input space}.

\paragraph*{Selection Gadget.}
We describe a gadget allowing us to model a discrete choice between~$\ell$ many possibilities (levees). We will describe the gadget centered at the origin of the $x_1$-$x_2$-plane. Later in our reduction, we will use several shifted versions of this gadget. An illustration of the selection gadget is given in \Cref{fig:selectiongadget}.

Each of the $\ell$ different choices corresponds to one of $\ell$ different slopes $s_1<s_2<\dots<s_\ell$. First, we place $13$ data points on the $x_2$-axis (with $x_1=0$, we call this vertical line $h_0$):

\begin{center}\small
	\begin{tabular}{c|ccccccccccccc}
		$x_2$&$-4$&$-3$&$-2$&$-\nicefrac{5}{3}$&$-\nicefrac{4}{3}$&$-1$&$0$&$1$&$\nicefrac{4}{3}$&$\nicefrac{5}{3}$&$2$&$3$&$4$\\\hline
		$y$&$0$&$0$&$0$&$\nicefrac{1}{3}$&$\nicefrac{2}{3}$&$1$&$1$&$1$&$\nicefrac{2}{3}$&$\nicefrac{1}{3}$&$0$&$0$&$0$
	\end{tabular}
\end{center}

Next, we need a small $\epsilon>0$ to be chosen later in a global context. The only condition we impose on $\epsilon$ in order to make the selection gadget work is that $\epsilon\leq\min\bigl\{\frac{1}{3\left| s_1 \right|},\frac{1}{3\left| s_\ell \right|}\bigr\}$. Based on this, we place $9$ data points parallel to the $x_2$-axis with $x_1=-\epsilon$ (we call the corresponding vertical line $h_{-\epsilon}$):

\begin{center}\small
	\begin{tabular}{c|ccccccccc}
		$x_2$&$-4-\epsilon s_\ell$&$-3-\epsilon s_\ell$&$-2-\epsilon s_\ell$&$-1-\epsilon s_1$&$0$&$1-\epsilon s_\ell$&$2-\epsilon s_1$&$3-\epsilon s_1$&$4-\epsilon s_1$\\\hline
		$y$&$0$&$0$&$0$&$1$&$1$&$1$&$0$&$0$&$0$
	\end{tabular}
\end{center}

Furthermore, similar to above, we place $9$ data points parallel to the $x_2$-axis with~$x_1=\epsilon$ (we call the corresponding line $h_\epsilon$):

\begin{center}\small
\begin{tabular}{c|ccccccccc}
	$x_2$&$-4+\epsilon s_1$&$-3+\epsilon s_1$&$-2+\epsilon s_1$&$-1+\epsilon s_\ell$&$0$&$1+\epsilon s_1$&$2+\epsilon s_\ell$&$3+\epsilon s_\ell$&$4+\epsilon s_\ell$\\\hline
	$y$&$0$&$0$&$0$&$1$&$1$&$1$&$0$&$0$&$0$
\end{tabular}
\end{center}

Finally, we place $2(\ell-1)$ many data points as follows: for each $i\in[\ell-1]$, we introduce one data point $\mathbf{q}_i^-\coloneqq (-\frac{4}{s_{i+1}-s_i},-\frac{2(s_i+s_{i+1})}{s_{i+1}-s_i})$, as well as one data point $\mathbf{q}_i^+\coloneqq(\frac{4}{s_{i+1}-s_i},\frac{2(s_i+s_{i+1})}{s_{i+1}-s_i})$. All these data points receive label $y=0$.

It is not too difficult to verify that a levee with slope $s_i$, $i\in[\ell]$, fits all data points of a selection gadget.
We omit the simple but tedious calculations here.
More intricately, the following lemma shows that a selection gadget indeed models a discrete choice between exactly $\ell$ possibilities.

\newcommand{\datapoint}[3]{
	\node[circle,inner sep=1.8,thick,draw,fill=black!#1] at (#2,#3) {};
}
\newcommand{\bottom}{00}
\newcommand{\low}{15}
\newcommand{\high}{30}
\newcommand{\topp}{45}
\newcommand{\dist}{1/3}

\newcommand{\myline}[3]{
	\draw[very thin,#3,black!60] #1 -- #2;
}

\newcommand{\fourlines}[2]{
	\myline{(-5,-2+#1)}{(5,-2-#1)}{#2}
	\myline{(-5,-1+#1)}{(5,-1-#1)}{#2}
	\myline{(-5,1+#1)}{(5,1-#1)}{#2}
	\myline{(-5,2+#1)}{(5,2-#1)}{#2}
}

\newcommand{\pointsandlines}{
	\fourlines{0}{dashed}
	\fourlines{5}{dashed}
	\fourlines{-5}{dashed}
	
	\datapoint{\bottom}{0}{-4}
	\datapoint{\bottom}{0}{-3}
	\datapoint{\bottom}{0}{-2}
	\datapoint{\low}{0}{-5/3}
	\datapoint{\high}{0}{-4/3}
	\datapoint{\topp}{0}{-1}
	\datapoint{\topp}{0}{0}
	\datapoint{\topp}{0}{1}
	\datapoint{\high}{0}{4/3}
	\datapoint{\low}{0}{5/3}
	\datapoint{\bottom}{0}{2}
	\datapoint{\bottom}{0}{3}
	\datapoint{\bottom}{0}{4}
	
	\foreach \xone in {-\dist,\dist}{
		\datapoint{\bottom}{\xone}{-4-\dist}
		\datapoint{\bottom}{\xone}{-3-\dist}
		\datapoint{\bottom}{\xone}{-2-\dist}
		\datapoint{\topp}{\xone}{-1+\dist}
		\datapoint{\topp}{\xone}{0}
		\datapoint{\topp}{\xone}{1-\dist}
		\datapoint{\bottom}{\xone}{2+\dist}
		\datapoint{\bottom}{\xone}{3+\dist}
		\datapoint{\bottom}{\xone}{4+\dist}
	}
	
	\datapoint{\bottom}{-4}{-2}
	\datapoint{\bottom}{-4}{2}
	\datapoint{\bottom}{4}{-2}
	\datapoint{\bottom}{4}{2}
}

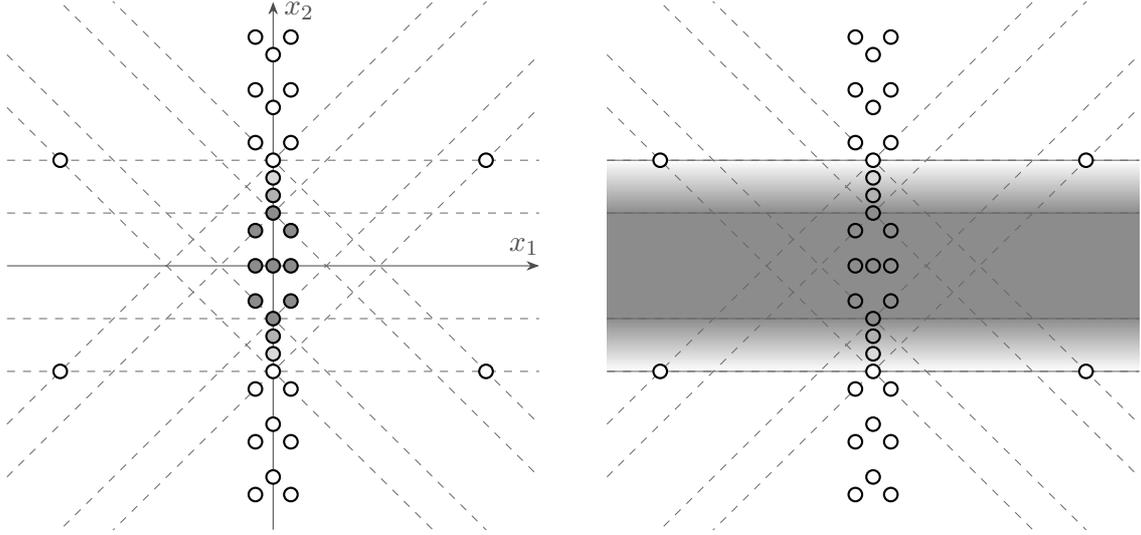
\begin{figure}[tb]
	\centering
	\hfill
	\begin{tikzpicture}[scale=0.7]
          \clip (-5,-5) rectangle (5,5);
          \small
          \draw[connection,color=black!70] (0,-5) -- (0,5);
          \node[anchor=north west, inner sep=0, color=black!70] at (0.2,5) {$x_2$};
          \draw[connection,color=black!70] (-5,0) -- (5,0); \node[anchor=south east, inner sep=0, color=black!70] at (5, 0.2){$x_1$};
          \pointsandlines
	\end{tikzpicture}\hfill\hfill
	\begin{tikzpicture}[scale=0.7]
          \clip (-5,-5) rectangle (5,5);
		\fill[top color=black!\topp, bottom color=black!\bottom] (-5,-2) rectangle (5,-1);
		\fill[color=black!\topp] (-5,-1) rectangle (5,1);
		\fill[top color=black!\bottom, bottom color=black!\topp] (-5,1) rectangle (5,2);
		\fourlines{0}{}
		\pointsandlines
	\end{tikzpicture}
	\hfill{}
	
	\caption{Illustration of the selection gadget with $\ell=3$ and $s_1=-1$, $s_2=0$, $s_3=1$. Both figures show the $x_1$-$x_2$-plane while the $y$-coordinate is indicated via the darkness of the gray color. The left picture shows all data points belonging to the gadget as well as the breaklines of the three possible levees fitting the data points. In addition to these features, the right picture shows a levee with slope $s_2=0$ as one of three possibilities to fit the data points of the gadget.}
	\label{fig:selectiongadget}
\end{figure}

\begin{lemma}\label{lem:selectiongadget}
  Let~$f\colon\R^2\to\R$ be a continuous piecewise linear function with only four breaklines
  that fits all the data points of the selection gadget.
  Then, $f=f_{s_i}$ for some $i\in[\ell]$.
\end{lemma}
\begin{proof}
	First, we focus on the three vertical lines $h_{-\epsilon}$, $h_0$, and $h_\epsilon$. Note each of the three lines contains a sequence of nine data points of which the first three have label $0$, the next three have label $1$ and the final three have label $0$ again. For simplicity, consider one of the three lines and denote these nine data points by $\mathbf p_1$ to $\mathbf p_9$. Note that $h_0$ contains even more data points, which will become important later. For the following argument, compare \Cref{fig:levee_crosssection}.
	
	\newcommand{\myaxes}{
		\draw[connection,color=black!70] (-5,0) -- (5,0) node[above] {$x_2$};
		\draw[connection,color=black!70] (0,-0.3) -- (0,2.5) node[right] {$y$};
	}

	\newcommand{\datapointlabelbelow}[4]{
		\datapoint{#1}{#2}{#3}
		\node at (#2,#3-.5) {#4};
	}

	\newcommand{\datapointlabelabove}[4]{
		\datapoint{#1}{#2}{#3}
		\node at (#2,#3+.5) {#4};
	}

	\newcommand{\ninepoints}{
		\datapointlabelbelow{\bottom}{-4}{0}{$\mathbf p_1$}
		\datapointlabelbelow{\bottom}{-3}{0}{$\mathbf p_2$}
		\datapointlabelbelow{\bottom}{-2}{0}{$\mathbf p_3$}
		\datapointlabelabove{\topp}{-1}{1}{$\mathbf p_4$}
		\datapointlabelabove{\topp}{0}{1}{$\mathbf p_5$}
		\datapointlabelabove{\topp}{1}{1}{$\mathbf p_6$}
		\datapointlabelbelow{\bottom}{2}{0}{$\mathbf p_7$}
		\datapointlabelbelow{\bottom}{3}{0}{$\mathbf p_8$}
		\datapointlabelbelow{\bottom}{4}{0}{$\mathbf p_9$}
	}
	
	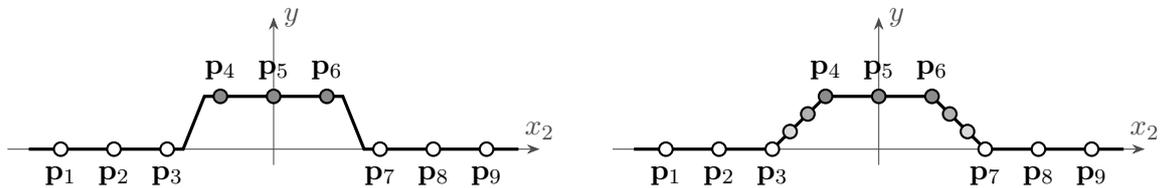
\begin{figure}
		\hfill
		\begin{tikzpicture}[scale=0.7]
			\small
			\myaxes
			\draw[very thick] (-4.6,0) -- (-1.7,0) -- (-1.3,1) -- (1.3,1) -- (1.7,0) -- (4.6,0);
			\ninepoints
		\end{tikzpicture}
		\hfill\hfill
		\begin{tikzpicture}[scale=0.7]
			\small
			\myaxes
			\draw[very thick] (-4.6,0) -- (-2,0) -- (-1,1) -- (1,1) -- (2,0) -- (4.6,0);
			\datapoint{\low}{-5/3}{1/3}
			\datapoint{\high}{-4/3}{2/3}
			\datapoint{\high}{4/3}{2/3}
			\datapoint{\low}{5/3}{1/3}
			\ninepoints
		\end{tikzpicture}
		\hfill{}
		
		\caption{Cross section of the selection gadget through one of the three lines $h_{-\epsilon}$, $h_0$, or $h_\epsilon$. The nine data points (labeled $\mathbf p_1$ to $\mathbf p_9$) on each of these lines force the function $f$ to attain a ``levee-shape'' with the exact position and slope of the ascending and descending sections as the only degrees of freedom (left). The four additional data points on $h_0$ even fix these properties and thus exactly determine $f$ on that line (right).}
		\label{fig:levee_crosssection}
	\end{figure}
	
	Observe that $f$ restricted to one of the three lines is a one-dimensional, continuous, piecewise linear function with at most four breakpoints. Looking at $\mathbf p_2$, $\mathbf p_3$, and $\mathbf p_4$, the corresponding $y$-labels are $0$, $0$, and $1$, respectively. This can only be fitted if there exists a convex breakpoint between $\mathbf p_2$ and $\mathbf p_4$. Analogously, there must be a concave breakpoint between $\mathbf p_3$ and $\mathbf p_5$, another concave breakpoint between $\mathbf p_5$ and $\mathbf p_7$, and a convex breakpoint between $\mathbf p_6$ and $\mathbf p_8$. This uses already all four available breakpoints, so there are no other breakpoints. Therefore, the function on the considered line must be linear outside the segment between $\mathbf p_2$ and $\mathbf p_8$. Since $\mathbf p_1$, $\mathbf p_2$, $\mathbf p_8$, and $\mathbf p_9$ all have label $0$, it follows that the function is constant $0$ outside this segment. Moreover, there is no concave breakpoint outside the segment between $\mathbf p_3$ and $\mathbf p_7$, implying that the function must be convex outside the segment between $\mathbf p_3$ and $\mathbf p_7$. However, since these two points have label $0$ as well, it follows that $f$ must even be constant $0$ there.
	
	Now consider the segment between $\mathbf p_4$ and $\mathbf p_6$. There is no convex breakpoint between $\mathbf p_4$ and $\mathbf p_6$. Therefore, the function must be concave within the segment. Since $\mathbf p_4$, $\mathbf p_5$, and $\mathbf p_6$ all have label $1$, it follows that the function is constant $1$ between $\mathbf p_4$ and $\mathbf p_6$.
	
	Putting together the insights gained so far, it follows that $f$ restricted to the considered line is constant $0$ first, goes up to constant $1$ via a convex and a concave breakpoint between $\mathbf p_3$ and $\mathbf p_4$, and goes down to constant $0$ again via a concave and a convex breakpoint between $\mathbf p_6$ and $\mathbf p_7$ (\Cref{fig:levee_crosssection}, left). Note that the exact location of these breakpoints and the slope in the sloped segments is not implied by the nine data points considered so far.
	
	This changes, however, when also taking into account the four other data points lying on $h_0$. Combined with the insights so far, they completely determine $f$ on this line (\Cref{fig:levee_crosssection}, right):
	\[
		f(0,x_2) = \left\{
			\begin{array}{ll}
				0,&\text{if } x_2\leq-2 \text{ or } x_2\geq2,\\
				1,&\text{if } -1\leq x_2 \leq 1,\\
				2+x_2,&\text{if } -2\leq x_2 \leq -1,\\
				2-x_2,&\text{if } 1\leq x_2 \leq 2.
			\end{array}
		\right.
	\]
	Observe that this is precisely the same as \eqref{eq:levee} with $x_1=0$.
	
	It remains to consider the behavior of $f$ on both sides of $h_0$. To this end, observe that the breakpoints of $f$ restricted to one of the three lines considered so far emerge as intersections of these lines with only four breaklines in total. Let us collect what we know so far about the locations of these four breaklines:
	\begin{itemize}
		\item There are exactly two convex breaklines, intersecting $h_0$ at $(0,-2)$ and $(0,2)$, respectively. We call them $g_1$ and $g_4$, respectively.
		\item There are exactly two concave breaklines, intersecting $h_0$ at $(0,-1)$ and $(0,1)$, respectively. We call them $g_2$ and $g_3$, respectively.
		\item Each of the four segments
		\begin{align*}
			I_1&\coloneqq[(-\epsilon, -2-\epsilon s_{\ell}), (-\epsilon, -1-\epsilon s_{1})]\subseteq [(-\epsilon,-\nicefrac{7}{3}),(-\epsilon,-\nicefrac{2}{3})],\\
			I_2&\coloneqq[(-\epsilon, 1-\epsilon s_{\ell}), (-\epsilon, 2-\epsilon s_{1})]\subseteq [(-\epsilon,\nicefrac{2}{3}),(-\epsilon,\nicefrac{7}{3})],\\
			I_3&\coloneqq[(\epsilon, -2+\epsilon s_{1}), (\epsilon, -1+\epsilon s_{\ell})]\subseteq [(\epsilon,-\nicefrac{7}{3}),(\epsilon,-\nicefrac{2}{3})],\text{ and}\\
			I_4&\coloneqq[(\epsilon, 1+\epsilon s_{1}), (\epsilon, 2+\epsilon s_{\ell})]\subseteq [(\epsilon,\nicefrac{2}{3}),(\epsilon,\nicefrac{7}{3})]
		\end{align*}
	is intersected by exactly one concave and one convex breakline. Here, the inclusions are implied by $\epsilon\leq\min\bigl\{\frac{1}{3\left| s_1 \right|},\frac{1}{3\left| s_\ell \right|}\bigr\}$. See \Cref{fig:segments} for an illustration of the position of these segments.
	\end{itemize}

	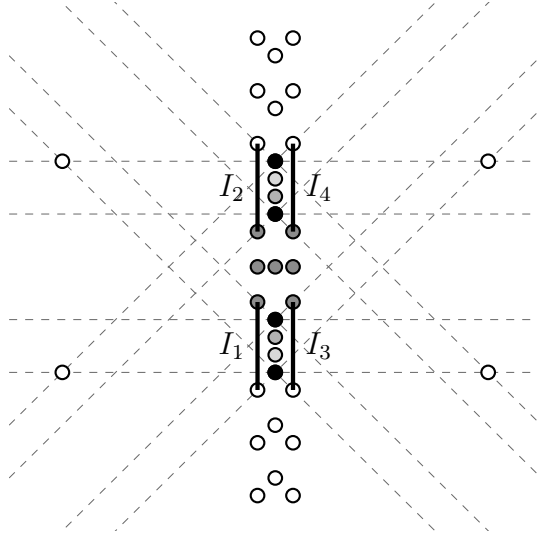
\begin{figure}[tb]
		\centering
		\begin{tikzpicture}[scale=0.7]
			\clip (-5,-5) rectangle (5,5);
			\begin{scope}[opacity=0.3]
				\pointsandlines
			\end{scope}
			\datapoint{100}{0}{-2}
			\datapoint{100}{0}{-1}
			\datapoint{100}{0}{1}
			\datapoint{100}{0}{2}
			\draw[ultra thick] (-1/3, -7/3) -- (-1/3, -2/3) node[midway, left] {$I_1$};
			\draw[ultra thick] (-1/3, 2/3) -- (-1/3, 7/3) node[midway, left] {$I_2$};
			\draw[ultra thick] (1/3, -7/3) -- (1/3, -2/3) node[midway, right] {$I_3$};
			\draw[ultra thick] (1/3, 2/3) -- (1/3, 7/3) node[midway, right] {$I_4$};
		\end{tikzpicture}
		\caption{Illustration of the segments $I_1$ to $I_4$ used in the proof of \Cref{lem:selectiongadget}. The figure also highlights (in black) the data points at $(0,-2)$ and $(0,2)$, each of which lies on a convex breakline, as well as the data points at $(0,-1)$ and $(0,1)$, each of which lies on a concave breakline.}
		\label{fig:segments}
	\end{figure}

	Now consider $g_2$, which goes through $(0,-1)$, and observe that it cannot intersect $I_2$ for the following reason. If it did, it would intersect $h_\epsilon$ at $x_2\leq-1-\nicefrac{5}{3}=-\nicefrac{8}{3}<-\nicefrac{7}{3}$ and hence would neither intersect $I_3$ nor $I_4$. This is a contradiction because there are only two concave breaklines and both $I_3$ and $I_4$ must be intersected by exactly one of them. Consequently, $g_2$ cannot intersect $I_2$, and must intersect $I_1$ instead.
	
	Analogously, it follows that $g_1$ and $g_2$ intersect $I_1$ and $I_3$. Similarly, $g_3$ and $g_4$ intersect $I_2$ and $I_4$. Combining this with the fact that $f$ restricted to each of the three vertical lines $h_{-\epsilon}$, $h_0$, and $h_\epsilon$ has an increasing section from $0$ to $1$ and a decreasing section from $1$ to $0$, this implies that the four lines $g_1$ to $g_4$ do not cross between $h_{-\epsilon}$ and $h_\epsilon$. Let us focus on the quadrilateral enclosed by $g_1$, $g_2$, $h_{-\epsilon}$ and $h_\epsilon$. By what we know so far, $f$ is constant $0$ on $g_1$, constant $1$ on $g_2$, and linear within this quadrilateral. Since $h_{-\epsilon}$ and $h_\epsilon$ are parallel, this implies that the corresponding two sides of the quadrilateral must have the same length. Thus, the quadrilateral must be a parallelogram. In particular, $g_1$ and $g_2$ are parallel. Similarly, $g_3$ and $g_4$ must be parallel.
	
	Let $s$ be the slope of $g_1$ and $g_2$, and let $t$ be the slope of $g_3$ and $g_4$. To complete the proof, we need to show that all four lines are parallel, that is, $s=t$, and that this slope value is equal to $s_i$ for some $i\in[\ell]$.
	
	Without loss of generality, we can assume that $s\leq t$, otherwise we mirror the gadget along the $x_2$-axis.
	Observe that $s_1\leq s\leq t\leq s_\ell$ because both $g_1$ and $g_2$ intersect $I_3$, and both~$g_3$ and $g_4$ intersect $I_4$.
	
	Let $i^*\coloneqq\max\{i\mid s_i \leq s\}\in[\ell]$. If $i^*=\ell$, then $s=t=s_\ell$ and we are done. Otherwise, consider the data point $\mathbf q_{i^*}^+= \left(\frac{4}{s_{i^*+1}-s_{i^*}},\frac{2(s_{i^*}+s_{i^*+1})}{s_{i^*+1}-s_{i^*}}\right)$, which has label $y=0$.
	
	Let us have a look at what $f$ restricted to the vertical line $h$ through $\mathbf q_{i^*}^+$ looks like. By $s\leq t$, the four lines $g_1$, $g_2$, $g_3$, and $g_4$ intersect $h$ in exactly this order (from bottom to top). This means that these lines do not cross between the $x_2$-axis and $h$. By our insights above, this implies that restricted to $h$, $f$ is zero outside the intersection points with $g_1$ and $g_4$, increases from zero to one between $g_1$ and $g_2$, stays constant $1$ between $g_2$ and $g_3$, and decreases back to $0$ between $g_3$ and $g_4$.
	
	By our choice of $i^*$, we obtain that $s_{i^*+1}>s$. Let us calculate at which $x_2$-coordinate~$g_1$ intersects $h$. This happens at
	\[x_2=-2+\frac{4s}{s_{i^*+1}-s_{i^*}}<-2+\frac{4s_{i^*+1}}{s_{i^*+1}-s_{i^*}}=\frac{4s_{i^*+1}-2(s_{i^*+1}-s_{i^*})}{s_{i^*+1}-s_{i^*}}=\frac{2(s_{i^*}+s_{i^*+1})}{s_{i^*+1}-s_{i^*}}.\]
        Thus, $\mathbf q_{i^*}^+$ lies strictly above the line $g_1$. Since $\mathbf q_{i^*}^+$ has label zero, this must imply that $\mathbf q_{i^*}^+$ does not lie below $g_4$. Looking at the intersection point of $g_4$ with $h$, this means:
	\begin{align*}
		&\quad2+\frac{4t}{s_{i^*+1}-s_{i^*}}\leq \frac{2(s_{i^*}+s_{i^*+1})}{s_{i^*+1}-s_{i^*}}\\
		\Leftrightarrow&\quad 2(s_{i^*+1}-s_{i^*})+4t\leq 2(s_{i^*}+s_{i^*+1})\\
		\Leftrightarrow&\quad t\leq s_{i^*}.
	\end{align*}

	Thus, we obtain $s_{i^*}\leq s \leq t \leq s_{i^*}$, implying that $g_1, g_2, g_3, g_4$ are all parallel and have one of the $\ell$ predefined slopes. This implies that $f$ is the levee $f_{s_{i^*}}$, completing the proof of the lemma.
\end{proof}

\paragraph*{Combining Multiple Selection Gadgets.}
Having constructed and understood a single selection gadget, the next step is to use multiple of these gadgets simultaneously. To this end, we will ``stack multiple selection gadgets upon each other along the $x_2$-axis''. To make this formal, we define a \emph{selection gadget with offset $z$} as the set of data points of a selection gadget as described above, where we add $z$ to all $x_2$-coordinates of the gadget. In other words, the gadget is centered around the point $(0,z)$.

Now, consider the set of data points originating from $m$ selection gadgets with offsets $z_1, \dots, z_m$, each one offering the choice between $\ell_j$ many slopes $s_i^{(j)}$, $i\in[\ell_j]$, $j\in[m]$. Suppose further that we uniformly choose $\epsilon\coloneqq\min_{j\in[m]}\min\Bigl\{\frac{1}{3\lvert s_1^{(j)} \rvert},\frac{1}{3\lvert s_\ell^{(j)} \rvert}\Bigr\}$ for all the gadgets such that the vertical lines $h_{-\epsilon}$, $h_0$, and $h_\epsilon$ with $x_1$-coordinates $-\epsilon$, $0$, and $\epsilon$, respectively, each contain either $9$ or $13$ data points from each gadget. Let $\delta\coloneqq\min_{j\in[m]}\min_{i\in[\ell_j-1]} (s_{i+1}^{(j)}- s_i^{(j)})$ be the smallest difference of two consecutive slopes in the $m$ gadgets. Moreover, let $S\coloneqq\max_{j\in[m]}\max_{i\in[\ell_j]} \lvert s_i^{(j)}\rvert$ be the largest absolute value of all the slopes. In this setting, the following lemma states that fitting all these data points is equivalent to independently choosing one slope for each single gadget and adding up the corresponding levees, provided that the distance of the gadgets is large enough.

\begin{lemma}\label{lem:manygadgets}
	If $z_{j+1}-z_j\geq \frac{8S}{\delta}+6$ for all $j\in[m-1]$, then there are exactly $\prod_{j=1}^{m} \ell_j$ many continuous piecewise linear functions $f\colon \R^2\to\R$ with at most $4m$ breaklines fitting the data points of the $m$ selection gadgets, namely $f(x_1,x_2)=\sum_{j=1}^{m} f_{s_{i_j}^{(j)}}(x_1, x_2-z_j)$ for each choice of indices $i_j\in[\ell_j]$ for each $j\in[m]$.
\end{lemma}
\begin{proof}
	We first show that each of these functions does indeed fit all the data points. For this, it is sufficient to show that each levee $f_{s_{i_j}^{(j)}}(x_1, x_2-z_j)$ is $0$ at all the data points $(\bar{x}_1,\bar{x}_2)$ belonging to a selection gadget with index $j'\neq j$. Without loss of generality, we can assume that $z_j=0$. By the definition of the selection gadget and checking all the possible $x_1$-coordinates, we obtain that $\lvert\bar{x}_1\rvert\leq\nicefrac{4}{\delta}$. Moreover, looking at the possible $x_2$-coordinates, we obtain that $\bar{x}_2$ can differ at most by $4+S\cdot\lvert\bar{x}_1\rvert$ from $z_{j'}$, from which we conclude $\lvert\bar{x}_2\rvert\geq \lvert z_{j'} \rvert - 4 - S\cdot\lvert\bar{x}_1\rvert\geq \frac{4S}{\delta}+2$. On the other hand, all points $(x_1,x_2)$ for which the levee $f_{s_{i_j}^{(j)}}(x_1, x_2)$ is nonzero satisfy $\lvert x_2 \rvert < 2 + \lvert s_{i_j} x_1 \rvert \leq 2 + S\lvert x_1 \rvert$. Since $\lvert\bar{x}_2\rvert\geq\frac{4S}{\delta}+2 \geq 2 + S\lvert x_1 \rvert$, it follows that $f_{s_{i_j}^{(j)}}$ must be zero at $(\bar{x}_1,\bar{x}_2)$, completing the proof that all claimed functions fit the $m$ selection gadgets.
	
	It remains to show that all functions $f$ fitting the data points of the $m$ selection gadgets are of the claimed form. We show this by induction on $m$. The base case $m=1$ is given by \Cref{lem:selectiongadget}. Now, let $m\geq 2$ and without loss of generality let $z_1=0$. We will again consider the three vertical lines $h_{-\epsilon}$, $h_0$, and $h_\epsilon$ with $x_1$-coordinates $-\epsilon$, $0$, and $\epsilon$, respectively. Remember that $f$ restricted to each of these three lines is a one-dimensional continuous piecewise linear function with at most $4m$ breakpoints, stemming from breaklines intersecting the respective vertical line. By looking at each individual gadget and arguing as in the proof of \Cref{lem:selectiongadget}, we obtain the following information:
	\begin{itemize}
		\item There are exactly $2m$ convex breaklines, intersecting $h_0$ at the $2m$ points $(0,z_j-2)$ and $(0,z_j+2)$, $j\in[m]$. Note that by our assumptions $z_1=0$ and $z_{j+1}-z_j\geq \frac{8S}{\delta}+6>6$, all these points are distinct, two of them are $(0,-2)$ and $(0,2)$, and all the other $2m-2$ points lie above the horizontal line $x_2=4$.
		\item There are exactly $2m$ concave breaklines, intersecting $h_0$ at the $2m$ points $(0,z_j-1)$ and $(0,z_j+1)$, $j\in[m]$. Again by our assumptions $z_1=0$ and $z_{j+1}-z_j\geq \frac{8S}{\delta}+6>6$, all these points are distinct, two of them are $(0,-1)$ and $(0,1)$, and all the other $2m-2$ points lie above the horizontal line $x_2=5$.
		\item Each of the four segments $I_1$ to $I_4$ corresponding to the selection gadget with index $j=1$ as defined in the proof \Cref{lem:selectiongadget} is intersected by exactly one convex and exactly one concave breakline. There are $4m-4$ further such segments stemming from selection gadgets with index $j>1$, and all of those lie completely above the horizontal line $x_2=6-\nicefrac{7}{3}=\nicefrac{11}{3}$.
	\end{itemize}

	Looking at the breaklines passing through $(0,-2)$ and $(0,-1)$, they must also pass through one of the described $2m$ segments on $h_{-\epsilon}$ and one of the described $2m$ segments on $h_\epsilon$. Since the considered gadget is the lowest one on the $x_2$-axis, the same argument as in the proof of \Cref{lem:selectiongadget} applies, which means that the only way of fulfilling these requirements simultaneously is that these breaklines pass through $I_1$ and $I_3$.
	Once having this, the same argument can be repeated for the breaklines passing through $(0,1)$ and $(0,2)$, making use of the fact that all the $4m-4$ segments not belonging to the considered gadget lie above the $x_2=\nicefrac{11}{3}$-line. Therefore, these breaklines must intersect $h_{-\epsilon}$ and $h_{\epsilon}$ within $I_2$ and $I_4$, respectively.
	
	From this, it follows as in the proof of \Cref{lem:selectiongadget} that the only way to fit the data points of the selection gadget with index $j=1$ is one of the $\ell_1$ levees $f_{s_{i}^{(1)}}$, $i\in[\ell_1]$. Thus, subtracting one of these $\ell_1$ levees from $f$ eliminates four of the $4m$ breaklines. Applying induction to the resulting function and the $m-1$ remaining selection gadgets completes the proof.
\end{proof}

\paragraph*{Global Construction.}

We are now ready to describe the overall reduction.
For a given formula~$F=C_1\wedge C_2\ldots\wedge C_m$ with variables~$v_1,\ldots,v_n$, we construct data points in $\R^2\times\R$ such that they can be fitted exactly with $k=4(m+n)$ ReLUs if and only if~$F$ is a yes-instance of \textsc{POITS}.
Our construction will consist of $m+n$ selection gadgets, namely one for each clause and one for each variable, and $3m$ further data points. Each of the $m$ selection gadgets corresponding to a clause determines which literal of this clause we choose to be true. Each of the $n$ selection gadgets corresponding to a variable determines whether this variable is true or false. The $3m$ remaining data points will ensure that these choices are consistent.
Let $\delta\coloneqq \frac{1}{2m}$. This will be the smallest difference of any two consecutive slopes in any selection gadget we are going to use. Moreover, no absolute value of a slope will be larger than $S\coloneqq 1$. From this, we conclude that, in order to apply \Cref{lem:manygadgets} in the end, we need to maintain a distance of at least $\Delta\coloneqq\frac{8S}{\delta}+6=16m+6$ between the centers of the gadgets.

We start by describing the positions and slopes of the selection gadgets. Compare \Cref{fig:global} for an illustration. Firstly, for each clause $C_j$, $j\in[m]$, we introduce one selection gadget with offset $j\Delta$ (that is, centered at $(0,j\Delta)$) and the three different slopes $s_1^{(j)}\coloneqq(2j-2)\delta - 1$, $s_2^{(j)}\coloneqq(2j-1)\delta-1$, and $s_3^{(j)}\coloneqq2j\delta-1$. Note that all these slopes are contained in $[-1,0]$. The interpretation will be as follows: Choosing the levee with slope $s_r^{(j)}$ for the $j$-th selection gadget corresponds to choosing the $r$-th literal of the $j$-th clause as the one that is set to true. Secondly, for each variable $v_i$, $i\in[n]$, we introduce one selection gadget with offset $-i\Delta$ and the two slopes $-1$ and $1$. Here the interpretation is as follows: choosing the levee with slope $-1$ corresponds to setting the variable to true, while choosing the levee with slope $1$ corresponds to setting the variable to false. Finally, if the $r$-th literal, $r\in[3]$, of clause $C_j$ is $v_i$, then we introduce a data point $\mathbf{p}_{j,r}$ with label $y=1$ at the intersection of the ``center-line'' of the levee with slope $s_r^{(j)}$ corresponding to the selection gadget for $C_j$ (that is, the line $x_2=\Delta j + s_r^{(j)}x_1$) and the ``center-line'' of the levee with slope $1$ corresponding to the selection gadget of $v_i$ (that is, the line $x_2=-\Delta i + x_1$). Thus, $\mathbf{p}_{j,r}\coloneqq (\frac{\Delta (i+j)}{1-s_r^{(j)}},\frac{\Delta (i+j)}{1-s_r^{(j)}}-\Delta i)$.

\newcommand{\linesandgadgets}{
	\footnotesize
	\clip(-.7,-5.7) rectangle (8.7,3.7);
	\foreach \jj in {1,2,3}{
		\foreach \rr in {0,1,2}{
			\draw[color=black!\topp] (-1, 2/3*\jj + 1/6*\rr + 1) -- (9, 4*\jj - 3/2*\rr - 9);
		}
		\node[rectangle,fill] at (0,\jj) {};
		\node at (.5,1.12*\jj) {$C_\jj$};
	}
	\foreach \ii in {1,...,5}{
		\draw[color=black!\topp] (-1,-\ii-1) -- (9, -\ii + 9);
		\draw[color=black!\topp] (-1,-\ii+1) -- (9, -\ii - 9);
		\node[rectangle,fill] at (0,-\ii) {};
		\node at (.5,-\ii) {$v_\ii$};
	}
}
\newcommand{\reductionpoints}{
	\draw[dotted, ultra thick] (0,-6) -- (0,4);
	\datapoint{\topp}{3}{-2}
	\datapoint{\topp}{30/11}{-4+30/11}
	\datapoint{\topp}{12/5}{-3+12/5}
	\datapoint{\topp}{18/5}{-4+18/5}
	\datapoint{\topp}{10/3}{-3+10/3}
	\datapoint{\topp}{3}{1}
	\datapoint{\topp}{6}{1}
	\datapoint{\topp}{30/7}{-2+30/7}
	\datapoint{\topp}{4}{3}
}

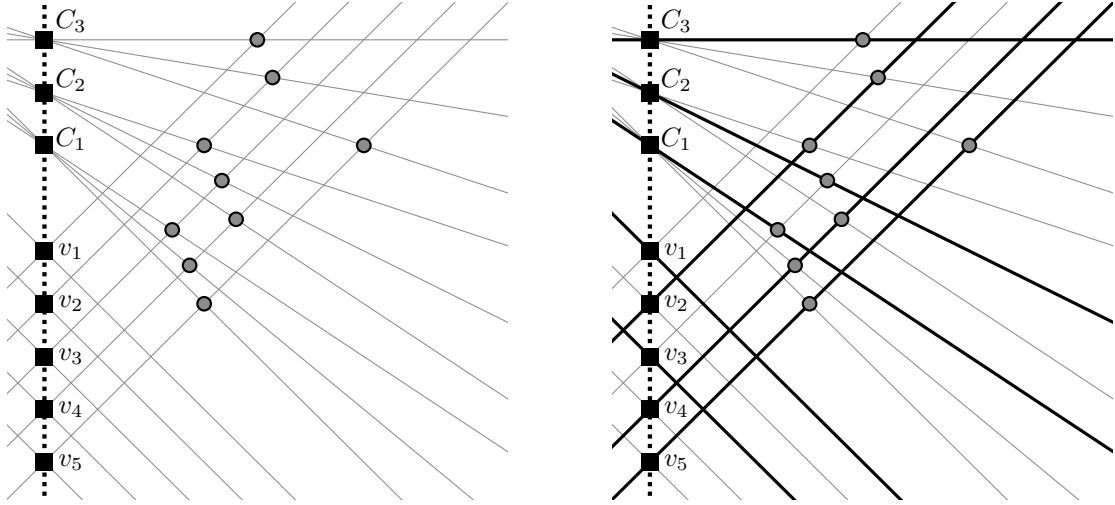
\begin{figure}[tb]
	\centering
	\hfill
	\begin{tikzpicture}[scale=0.7]
		\linesandgadgets
		\reductionpoints
	\end{tikzpicture}
	\hfill\hfill
	\begin{tikzpicture}[scale=0.7]
		\linesandgadgets
		\draw[very thick] (-1,3) -- (9,3);
		\draw[very thick] (-1,2.5) -- (9,-2.5);
		\draw[very thick] (-1,5/3) -- (9,-5);
		\draw[very thick] (-1,0) -- (9,-10);
		\draw[very thick] (-1,-3) -- (9,7);
		\draw[very thick] (-1,-2) -- (9,-12);
		\draw[very thick] (-1,-5) -- (9,5);
		\draw[very thick] (-1,-6) -- (9,4);
		\reductionpoints
	\end{tikzpicture}
	\hfill{}
	\caption{Global construction layout for the reduction from \textsc{POITS} to \RELU{\mathcal{L}}. The figure shows the construction for the instance $(v_5\vee v_4 \vee v_3)\wedge (v_4\vee v_3 \vee v_2)\wedge (v_5\vee v_2 \vee v_1)$. The vertical dotted line is the $x_2$-axis along which we place all the selection gadgets. Each gadget is depicted with a black square. Each solid gray line depicts one possible levee. Each gray circle depicts a data point $\mathbf p_{j,r}$ with label one. The picture on the right additionally shows one possible solution to the given instance. Indeed, choosing levees corresponding to the solid black lines selects exactly one levee per selection gadget and exactly one levee passing through each of the nine additional data points. This corresponds to the truth assignment $v_1=v_3=\text{true}$ and $v_2=v_4=v_5=\text{false}$.}
	\label{fig:global}
\end{figure}

This finishes the construction. Before we prove \Cref{thm:d=2} using this construction, we show the following useful lemma.

\begin{lemma}\label{lem:nonzerolevees}
	For each $j\in[m]$ and $r\in[3]$, there are exactly two out of the $3m+2n$ possible levees defined by the selection gadgets which are non-zero at $\mathbf p_{j,r}$, namely $f_{s_{r}^{(j)}}(x_1, x_2-j\Delta)$ and $f_1(x_1, x_2+i\Delta)$, where $v_i$ is the $r$-th literal in $C_j$.
\end{lemma}
\begin{proof}
	Since $\mathbf p_{j,r}$ is the intersection point of the center-lines of the two named levees, it suffices to show that no other levee is non-zero at this point.
	
	Let us start by reminding ourselves that a levee with offset $z$ and slope $s$ is non-zero only for points within a stripe of ``vertical width 4'', that is, for points $(x_1,x_2)$ with $sx_1+z-2< x_2 <sx_1+z+2$.
	
	Now we focus on levees belonging to other clauses $C_{j'}$ with $j'\neq j$. If $j'>j$, then the slope will be at least $s_{r}^{(j)}$ and the offset will be at least $(j+1)\Delta$. Since $\mathbf p_{j,r}$ lies on the right-hand side of the $x_2$-axis and on the center-line of a levee with slope exactly $s_{r}^{(j)}$ and offset exactly $j\Delta$, we obtain that $\mathbf p_{j,r}$ lies below the center-line of the considered levee with a vertical distance of at least $\Delta>2$, implying that the levee must vanish at $\mathbf p_{j,r}$. In the case $j'<j$ it follows similarly with $\mathbf p_{j,r}$ lying above instead of below the considered levee.
	
	Next, let us focus on the two levees belonging to the same clause $C_j$ but to the $r'$-th literal with $r'\neq r$. The slope of such a levee differs by at least $\delta$ from $s_{r}^{(j)}$, while the offset is exactly $j\Delta$. This implies that $\mathbf p_{j,r}$ has a vertical distance of at least $\delta\frac{\Delta (i+j)}{1-s_r^{(j)}}\geq\delta\frac{2\Delta}{2}=\delta\Delta>8>2$ from the center-line of the considered levee.
	
	Next, let us focus on a levee with slope $1$ belonging to a variable $v_{i'}$ with $i'\neq i$. Since $\mathbf p_{j,r}$ lies on the center-line of the levee with slope $1$ belonging to $v_i$, these levees are parallel, and have vertical distance at least $\Delta>2$, this case is settled, too.
	
	Finally, let us focus on a levee  with slope $-1$ belonging to any variable. Such a levee has an offset of at most $-\Delta$ and its slope is at most $s_{r}^{(j)}$. Since $\mathbf p_{j,r}$ lies on the center-line of the levee with offset $j\Delta$ and slope $s_{r}^{(j)}$, this implies that its vertical distance to the considered levee is at least $2\Delta>2$, finishing the proof.
\end{proof}

Finally, we are ready to prove the main theorem.

\begin{proof}[Proof of \Cref{thm:d=2}]
	We reduce from \textsc{POITS} and construct an instance of \RELU{\mathcal{L}} with $k=4(m+n)$ and $\gamma=0$ as described above. Note that, overall, we introduce~$O(m+n)$ points with rational coordinates (with~$\poly(m,n)$ bits) which are polynomial-time computable.
	
	To prove equivalence between the \textsc{POITS} instance and the constructed  instance, let us first assume that the \textsc{POITS} instance is a yes-instance. Let $T\subseteq[n]$ be a set of indices such that the truth assignment with $v_i=\text{true}$ for $i\in T$ and $v_i=\text{false}$ for $i\notin T$ sets exactly one literal per clause to true. Let $r_j\in\{1,2,3\}$ denote which of the three literals is set to true in clause $C_j$ by this assignment. We claim that the following function, which is a sum of $m+n$ levees and thus realizable with $k=4(m+n)$ ReLUs using \Cref{obs:levee}, exactly fits all the constructed data points:
	\begin{equation}\label{eq:leveesum}
		f(x_1,x_2) =
			\sum_{i\in T} f_{-1} (x_1, x_2+i\Delta)
		+	\sum_{i\notin T} f_{1} (x_1, x_2+i\Delta)
		+	\sum_{j=1}^{m} f_{s_{r_j}^{(j)}}(x_1, x_2-j\Delta).
	\end{equation}
	By \Cref{lem:manygadgets}, $f$ fits all data points belonging to the selection gadgets. It remains to show that $f$ attains value $1$ at all the data points $\mathbf p_{j,r}$, $j\in[m]$, $r\in[3]$. To see this, fix such $j$ and $r$ and let the $r$-th literal in $C_j$ be $v_i$. By \Cref{lem:nonzerolevees}, the only two levees which can potentially be non-zero at $\mathbf p_{j,r}$ are $f_{s_{r}^{(j)}}(x_1, x_2-j\Delta)$ and $f_1(x_1, x_2+i\Delta)$. If $r=r_j$, then $v_i=\text{true}$ and the former levee attains value $1$ while the latter levee attains value $0$ at $\mathbf p_{j,r}$. Otherwise, if $r\neq r_j$, then $v_i=\text{false}$ and the former levee attains value $0$ while the latter levee attains value $1$ at $\mathbf p_{j,r}$. In both cases, the data point is fitted correctly.
	
	Now suppose conversely that the constructed data points can be precisely fitted with a function $f$ representable with $k=4(m+n)$ ReLUs. By \Cref{lem:manygadgets}, $f$ must be of the form \eqref{eq:leveesum} for some set $T\subseteq[n]$ and some values $r_j\in[3]$ for all $j\in[m]$. We claim that setting $v_i=\text{true}$ for $i\in T$ and $v_i=\text{false}$ for $i\notin T$ sets exactly one literal per clause to true. To see this, fix $j\in[m]$ and $r\in[3]$ and let $v_i$ be the $r$-th literal of~$C_j$. Using \Cref{lem:nonzerolevees} again, observe that exactly one of the two levees $f_{s_{r}^{(j)}}(x_1, x_2-j\Delta)$ and $f_1(x_1, x_2+i\Delta)$ must belong to the sum \eqref{eq:leveesum} because the data point $\mathbf p_{j,r}$ has label one.
	In other words, it holds that either $r=r_j$ (implying $i\in T$) or~$i\not\in T$.
	This implies that, for each $j\in[m]$, the defined truth assignment sets exactly the $r_j$-th literal of $C_j$ to true, finishing the overall proof.
\end{proof}

\section{W[1]-Hardness for Four ReLUs}

We show that fixed-parameter tractability with respect to~$d$ is unlikely even for target error zero and four ReLUs. In fact, we prove a running time lower bound of~$n^{\Omega(d)}$ based on the~ETH.

\begin{theorem}\label{thm:k=4}
  \RELU{\mathcal{L}} with~$k=4$ and $\gamma=0$ is W[1]-hard with respect to~$d$ and not solvable in $\rho(d)n^{o(d)}\poly(L)$ time (where~$L$ is the input bit-length) for any function~$\rho$ assuming the ETH.
\end{theorem}

We prove \Cref{thm:k=4} with a parameterized reduction from the \textsc{2-Hyperplane Separability} problem.

\problemdef{2-Hyperplane Separability}
{Two point sets~$Q$ and~$P$ in $\R^d$.}
{Are there two hyperplanes that strictly separate~$Q$ and~$P$?}

Here, two hyperplanes \emph{strictly separate}~$Q$ and~$P$ if, for every pair~$(\mathbf{q},\mathbf{p})\in Q\times P$, the open line segment~$\mathbf{qp}$ is intersected by at least one hyperplane and no point from~$Q\cup P$ is contained in any of the two hyperplanes.
\citet{GKR09} showed that this problem is W[1]-hard with respect to~$d$ and not solvable in $\rho(d)m^{o(d)}\poly(L)$ time assuming the ETH (where~$m\coloneqq |Q\cup P|$ and~$L$ is the instance size).
In fact, their proof shows that if there is a solution, then there is a solution where~$Q$ lies entirely in one region of
the hyperplane arrangement and the points in~$P$ lie only in the two neighboring regions.
Formally, if the two hyperplanes are defined by~$\mathbf{h}_i \cdot \mathbf{x}+o_i=0$ for $\mathbf{h}_i\in\R^d$, $o_i\in\R$, $i\in[2]$, then (without loss of generality) we can assume that the following holds:
\begin{align}
  &\forall \mathbf{q}\in Q : \mathbf{h}_1 \cdot \mathbf{q} + o_1 > 0 > \mathbf{h}_2 \cdot \mathbf{q} + o_2\label{eqn:q}\\
  &\forall \mathbf{p}\in P : \sgn(\mathbf{h}_1 \cdot \mathbf{p} + o_1) = \sgn(\mathbf{h}_2 \cdot \mathbf{p} + o_2)\label{eqn:p}
\end{align}
Moreover, a closer inspection of their reduction shows that one can assume that the hyperplanes have distance at least $\epsilon\coloneqq m^{-3}$ to each input point\footnote{The critical points in the reduction are the constraint points~$q_{ij}^{uv}$ which are separated from the points $p_{iu_i},p_{i\overline{u}_i},p_{ju_j},p_{j\overline{u}_j}$ by some translation of the hyperplane $H(u_1,\ldots,u_k)$ towards the origin.
  The distance of any $q_{ij}^{uv}$ to~$H(u_1,\ldots,u_k)$ is at least~$2\sin^3(\pi/m)\ge 2m^{-3}$ in one dimension.}. That is, we can assume
\begin{align}
  \forall \mathbf{x}\in Q\cup P, i\in[2] : \frac{|\mathbf{h}_i\cdot \mathbf{x} +o_i|}{\Vert\mathbf{h}_i\Vert} > \epsilon.\label{eqn:eps}
\end{align}
We will make use of these assumptions in the following proof.

\begin{proof}[Proof of \Cref{thm:k=4}]
  Let~$(Q,P)$ be an instance of (restricted)~\textsc{2-Hyperplane Separability} and let~$m\coloneqq|Q\cup P|$ and~$\epsilon\coloneqq m^{-3}$.
  We construct the instance~$(X\subseteq \R^{d+1},k\coloneqq 4,\gamma\coloneqq 0)$ of \RELU{\mathcal{L}}, where~$X$ contains the following points:
  \begin{itemize}
    \item $(\mathbf{q},1)$ for each $\mathbf{q}\in Q$,
    \item $(\mathbf{p},0)$ for each $\mathbf{p}\in P$,
    \item $(\mathbf{r}_\mathbf{qp}\coloneqq (1-\delta)\mathbf{q}+\delta \mathbf{p} ,1)$ and~$(\mathbf{s}_\mathbf{qp}\coloneqq \delta\mathbf{q}+(1-\delta)\mathbf{p},0)$ for each $(\mathbf{q},\mathbf{p})\in Q\times P$,
      where $\delta\coloneqq \epsilon(2\Vert\mathbf{q}-\mathbf{p}\Vert)^{-1}$.
    \end{itemize}
    Note that $\mathbf{r}_\mathbf{qp}$ ($\mathbf{s}_\mathbf{qp}$) lies on the line segment~$\mathbf{qp}$ at distance~$\nicefrac\epsilon 2$ to~$\mathbf{q}$ ($\mathbf{p}$).
  Overall, we construct $n\coloneqq|X|\in O(m^2)$ points, which can be done in polynomial time.

  For the correctness, assume first that there are two hyperplanes~$\mathcal{H}_i$, $i\in[2]$, defined by~$\mathbf{h}_i\cdot \mathbf{x}+ o_i=0$ (wlog $\Vert \mathbf{h}_i\Vert=1$) that strictly separate~$Q$ and~$P$ and satisfy~(\ref{eqn:q})--(\ref{eqn:eps}).
  
  A solution for~$(X,4,0)$ can then be constructed as follows (see also \Cref{fig:2HyperplaneExample}):
  We use two ReLUs realizing an ``upward step'' of height~1 (with slope~$\beta\coloneqq \nicefrac{4}{\epsilon}$) in the direction of~$\mathbf{h}_1$.
  That is, we set
  \begin{align*}
    \mathbf{w}_1&\coloneqq \beta\mathbf{h}_1, &&b_1\coloneqq \beta o_1, &&a_1\coloneqq 1,\\
    \mathbf{w}_2&\coloneqq \beta\mathbf{h}_1, &&b_2\coloneqq \beta o_1-1, &&a_2\coloneqq -1.
  \end{align*}
  \noindent
  Additionally, we use two ReLUs realizing a ``downward step'' of height~1 (with slope~$-\beta$) in the direction of~$\mathbf{h}_2$, that is,
  \begin{align*}
    \mathbf{w}_3&\coloneqq \beta\mathbf{h}_2, &&b_3\coloneqq \beta o_2, &&a_3\coloneqq -1,\\
    \mathbf{w}_4&\coloneqq \beta\mathbf{h}_2, &&b_4\coloneqq \beta o_2-1, &&a_4\coloneqq 1.
  \end{align*}

  Let~$\mathcal{W}_i$ be the hyperplane defined by~$\mathbf{w}_i\cdot\mathbf{x}+b_i=0$ for~$i\in[4]$.
  Note that $\mathcal{W}_1=\mathcal{H}_1$ and~$\mathcal{W}_3=\mathcal{H}_2$.
  Note further that~$\mathcal{W}_2$ is parallel to~$\mathcal{W}_1$ at distance~$\beta^{-1}=\nicefrac{\epsilon}{4}$ and~$\mathcal{W}_4$ is parallel to~$\mathcal{W}_3$ at distance~$\nicefrac{\epsilon}{4}$.
  
  To verify that all data points are exactly fitted, consider first a point~$\mathbf{q}\in Q$.
  From~(\ref{eqn:q}) and~(\ref{eqn:eps}), we obtain
  \begin{align*}
    &\mathbf{w_1}\cdot\mathbf{q}+b_1 = \beta(\mathbf{h}_1\cdot \mathbf{q} +o_1) > 0,\\
    &\mathbf{w_2}\cdot\mathbf{q}+b_2 = \beta(\mathbf{h}_1\cdot \mathbf{q} +o_1-\beta^{-1}) > \beta\epsilon-1 > 0,\\
    &\mathbf{w_3}\cdot\mathbf{q}+b_3 = \beta(\mathbf{h}_2\cdot \mathbf{q} +o_2) < 0,\\
    &\mathbf{w_4}\cdot\mathbf{q}+b_4 = \beta(\mathbf{h}_2\cdot \mathbf{q} +o_2-\beta^{-1}) < 0.
  \end{align*}
  From the above inequalities, it follows
  \begin{align*}
    \phi(\mathbf{q})
    =\beta(\mathbf{h}_1\cdot\mathbf{q}+o_1)-\beta(\mathbf{h}_1\cdot\mathbf{q}+o_1-\beta^{-1}) =1.
  \end{align*}
  Now consider a point~$\mathbf{r}_\mathbf{qp}$ and note that, for each~$\mathcal{W}_i$, $\mathbf{r}_\mathbf{qp}$ lies in the same half-space as~$\mathbf{q}$ since it has distance~$\nicefrac{\epsilon}{2}$ to~$\mathbf{q}$ which has distance at least~$\frac{3}{4}\epsilon$ to~$\mathcal{W}_i$ (by (\ref{eqn:eps})).
  Thus,
  \begin{align*}
    \phi(\mathbf{r}_\mathbf{qp}) =\beta(\mathbf{h}_1\cdot\mathbf{r}_\mathbf{qp}+o_1)-\beta(\mathbf{h}_1\cdot\mathbf{r}_\mathbf{qp}+o_1-\beta^{-1}) =1.
  \end{align*}
  
  Next, consider a point~$\mathbf{p}\in P$. Using (\ref{eqn:p}) and (\ref{eqn:eps}), one easily verifies that
  \begin{align*}
   \sgn(\mathbf{w_1}\cdot\mathbf{p}+b_1) = \sgn(\mathbf{w_2}\cdot\mathbf{p}+b_2) = \sgn(\mathbf{w_3}\cdot\mathbf{p}+b_3) = \sgn(\mathbf{w_4}\cdot\mathbf{p}+b_4).
  \end{align*}
  Hence, $\phi(\mathbf{p})=0$ clearly holds if all the above signs are negative. If all signs are positive, then
  \[\phi(\mathbf{p})=\beta(\mathbf{h}_1\cdot \mathbf{p} +o_1) - \beta(\mathbf{h}_1\cdot \mathbf{p} +o_1-\beta^{-1})-\beta(\mathbf{h}_2\cdot \mathbf{p} +o_2)+\beta(\mathbf{h}_2\cdot \mathbf{p} +o_2-\beta^{-1})\\
    =0.\]
  Finally, any point~$\mathbf{s}_\mathbf{qp}$ analogously lies in the same half-space as~$\mathbf{p}$ for each~$\mathcal{W}_i$, which also implies~$\phi(\mathbf{s}_\mathbf{qp})=0$.
  Thus, all points are correctly fitted.
  
  Conversely, assume that the points in~$X$ can be exactly fitted by~$\phi$ realized by four ReLUs with values $\mathbf{w}_i,b_i,a_i$, $i\in[4]$.
  Let~$I^+\coloneqq \{i\in[4]\mid a_i=1\}$ and $I^-\coloneqq \{i\in[4]\mid a_i=-1\}$.
  
  Consider an arbitrary line segment $\mathbf{qp}$ for $(\mathbf{q},\mathbf{p})\in Q\times P$.
  Clearly, the points~$(\mathbf{q},1)$, $(\mathbf{r}_\mathbf{qp},1)$ and~$(\mathbf{s}_\mathbf{qp},0)$ on this line segment cannot all lie on the same piece of~$\phi$.
  Hence, $\phi$ must have a concave breakpoint at some point on the open segment between~$\mathbf{q}$ and~$\mathbf{p}$.
  That is, there must be a ReLU~$i\in I^-$ such that the hyperplane defined by~$(\mathbf{w}_i,b_i)$ intersects the open line segment $\mathbf{qp}$ and does not contain~$\mathbf{q}$ or~$\mathbf{p}$.
  Analogously, the points $(\mathbf{p},0)$, $(\mathbf{s}_\mathbf{qp},0)$ and~$(\mathbf{r}_\mathbf{qp},1)$ enforce a convex breakpoint, that is, a ReLU $j\in I^+$ with a hyperplane~$(\mathbf{w}_j,b_j)$ also intersecting the open line segment $\mathbf{qp}$ and not containing~$\mathbf{q}$ or~$\mathbf{p}$.
  
  To sum up, every open line segment~$\mathbf{q}\mathbf{p}$ is intersected by at least two hyperplanes (not containing~$\mathbf{q}$ or~$\mathbf{p}$), one corresponding to a ReLU $i\in I^{-1}$ and one corresponding to a ReLU $j\in I^+$.
  Since there are only four ReLUs, it follows that $\min(|I^+|,|I^-|)\le 2$.
  That is, we obtain a solution for \textsc{2-Hyperplane Separability} by picking either all hyperplanes corresponding to~$I^+$ or all hyperplanes corresponding to~$I^-$.
  
  This finishes the reduction.
  Note that since the dimension of the input data points in our constructed instance is~$d$,
  any algorithm solving \RELU{\mathcal{L}} in time~$\rho(d)n^{o(d)}\poly(L)$ would imply an algorithm running in time~$\rho(d)m^{o(d)}\poly(L')$ for \textsc{2-Hyperplane Separability} contradicting the ETH.
\end{proof}

\newcommand{\auxpoint}[3]{
	\node[circle,inner sep=1.4,draw,fill=black!#1] at (#2,#3) {};
}

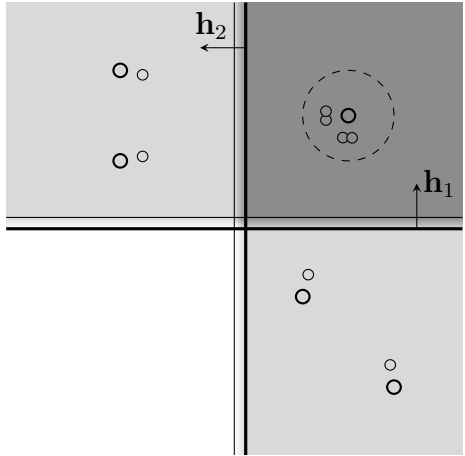
\begin{figure}[t]
  \centering
  \begin{tikzpicture}[scale=.6]
    \def\eps{1};
    
    \fill[top color=black!\topp, bottom color=black!\low] (\eps/4,0) rectangle (5,\eps/4);
    \fill[top color=black!\low, bottom color=black!\bottom] (0,0) rectangle (-5,\eps/4);
    \fill[left color=black!\bottom, right color=black!\low] (0,-5) rectangle (\eps/4,0);
    \fill[left color=black!\low, right color=black!\topp] (0,\eps/4) rectangle (\eps/4,5);
    \fill[top color=black!\topp, bottom color=black!\bottom, shading angle=-45] (0,0) rectangle (\eps/4,\eps/4);
    \fill[color=black!\topp] (\eps/4,\eps/4) rectangle (5,5);
    \fill[color=black!\bottom] (0,0) rectangle (-5,-5);
    \fill[color=black!\low] (\eps/4,-5) rectangle (5,0);
    \fill[color=black!\low] (0,\eps/4) rectangle (-5,5);
    
    \draw (0,-5)--(0,5);
    \draw[very thick] (\eps/4,-5)--(\eps/4,5);
    \draw[very thick] (-5,0)--(5,0);
    \draw (-5,\eps/4)--(5,\eps/4);

    \draw[-stealth] (\eps/4,4) -- (\eps/4-1,4);
    \node (w1) at (-0.5,4.5) {$\mathbf{h}_2$};
    \draw[-stealth] (4,0) -- (4,1);
    \node (w1) at (4.5,1) {$\mathbf{h}_1$};

    \datapoint{\topp}{2.5}{2.5};
    \draw[dashed] (2.5,2.5) circle (\eps);
    \auxpoint{\topp}{2.5-0.242\eps/2}{2.5-0.97\eps/2};
    \auxpoint{\topp}{2.5+0.164\eps/2}{2.5-0.98\eps/2};
    \auxpoint{\topp}{2.5-0.98\eps/2}{2.5-0.196\eps/2};
    \auxpoint{\topp}{2.5-0.98\eps/2}{2.5+0.196\eps/2};

    \datapoint{\low}{1.5}{-1.5};
    \auxpoint{\low}{1.5+0.242\eps/2}{-1.5+0.97\eps/2};
    
    \datapoint{\low}{3.5}{-3.5};
    \auxpoint{\low}{3.5-0.164\eps/2}{-3.5+0.98\eps/2};
    
    \datapoint{\low}{-2.5}{1.5};
    \auxpoint{\low}{-2.5+0.98\eps/2}{1.5+0.196\eps/2};
    
    \datapoint{\low}{-2.5}{3.5};
    \auxpoint{\low}{-2.5+0.98\eps/2}{3.5-0.196\eps/2};
  \end{tikzpicture}
  \caption{Example of the reduction from \textsc{2-Hyperplane Separability} for~$d=2$ dimensions.
    Big points are points in~$Q$ (dark gray) and in~$P$ (light gray). The small points are additionally introduced. The four lines are the breaklines of the four ReLUs. The two thick lines indicate the original two separating lines. The dashed circle has radius~$\epsilon$.}
  \label{fig:2HyperplaneExample}
\end{figure}

\section{Hardness Results for Linear Threshold Activations}\label{sec:LT}

A nowadays less popular, but more classical activation function than ReLU is the linear threshold function $x\mapsto\mathds{1}_{\{x>0\}}$. Analogously to \RELU{\mathcal{L}}, we consider the following decision version of the training problem for linear threshold functions:

\problemdef{\LT{\mathcal{L}}}
{Data points~$(\mathbf{x}_1,y_1),\ldots,(\mathbf{x}_n,y_n)\in \R^d\times\R$, a number~$k\in\N$ of linear threshold neurons, and a target error~$\gamma\in\R_{\ge 0}$.}
{Are there weights $\mathbf w_1, \ldots, \mathbf w_k\in \R^d$, biases $b_1,\ldots,b_k\in\R$, and coefficients $a_1,\ldots,a_k\in\R$ such that
	\[
	\sum_{i=1}^n \mathcal{L}\Bigg(\sum_{j=1}^k a_j\mathds{1}_{\{\mathbf{w}_j\cdot\mathbf{x}_i + b_j>0\}},\ y_i\Bigg) \le \gamma?
	\]
}

Note that for linear thresholds, we cannot assume $a_j\in\{-1,1\}$ because the normalization used in the ReLU case does not apply here.

As in the ReLU case, the crucial ingredient to study the training complexity of linear threshold networks is their geometry. To this end, observe that every function represented by a 2-layer linear threshold network is piecewise constant, where the pieces emerge from the hyperplane arrangement defined by the~$k$ hyperplanes $\mathbf w_j\cdot \mathbf x + b_j = 0$, $j\in[k]$, corresponding to the hidden neurons. Since our reductions for the ReLU case always use two ReLUs to approximate ``step functions'' from $0$ to $1$ and from~$1$ to~$0$, it is easy to adapt the reductions to the linear threshold case.

\begin{corollary}\label{cor:d=2}
	\LT{\mathcal{L}} is NP-hard even for $d=2$ and $\gamma=0$.
\end{corollary}
\begin{proof}
	We use an analogous reduction to the proof of \Cref{thm:d=2}. Instead of a sum of levees, we use a sum of ``stripes'' within which the function attains value $1$. With this idea, it is straight-forward to build selection gadgets and an analogous global construction. Note that the number $k$ of required linear threshold neurons is only $k=2(m+n)$ for a \textsc{POITS} instance with $m$ clauses and $n$ variables because each stripe can be realized with two linear threshold neurons instead of the four ReLUs required to build a levee.
\end{proof}

For the sake of completeness, we note that also the W[1]-hardness result by \citet{FHN22} extends to linear threshold functions. To this end, consider the $\ell^p$-loss $\ell^p(\hat{y},y)=\lvert \hat{y}-y \rvert^p$, with $\ell^0$ simply counting the non-zero components of $\hat{y}-y$.

\begin{corollary}\label{cor:k=1}
	For each $p\in\left[0,\infty\right[$, \LT{\ell^p} with $k=1$ is NP-hard, W[1]-hard with respect to $d$ and not solvable in $\rho(d)n^{o(d)}\poly(L)$ time (where $L$ is the input bit-length) for any function $\rho$ assuming the ETH.
\end{corollary}
\begin{proof}
	Having a careful look into the reduction from \textsc{Multicolored Clique} by \citet{FHN22}, it turns out that the single ReLU neuron used in this reduction can be replaced by a linear threshold neuron without changing the logic of the reduction.
\end{proof}

Finally, also \Cref{thm:k=4} finds its analogue for the linear threshold case.

\begin{corollary}\label{cor:k=2}
	\LT{\mathcal{L}} with~$k=2$ and $\gamma=0$ is W[1]-hard with respect to~$d$ and not solvable in $\rho(d)n^{o(d)}\poly(L)$ time (where~$L$ is the input bit-length) for any function~$\rho$ assuming the ETH.
\end{corollary}
\begin{proof}
	The proof is analogous to (even much easier than) the one of \Cref{thm:k=4}. Instead of two ReLUs  to realize a step of height one, we can simply use one linear threshold neuron (which is why we obtain hardness already for $k=2$ in this case). Note that we do not even need to introduce the additional data points $\mathbf r_{\mathbf q \mathbf p}$ and  $\mathbf s_{\mathbf q \mathbf p}$ and obtain a much more direct reduction from \textsc{2-Hyperplane Separability}.
\end{proof}

\section{An Algorithm for Exact Fitting in the Convex Case}

Contrasting the previous hardness results for \RELU{\mathcal{L}}, we now consider the tractable special case where all coefficients~$a_j$ are 1.
In this case, the neural network realizes a convex continuous piecewise linear function
$\phi(\mathbf{x})=\sum_{j=1}^k[\mathbf{w}_j\cdot\mathbf{x}+b_j]_+$
with at most~$2^k$ distinct (affine) pieces.
We show that this case (which we call $\RELU{\mathcal{L}}^+$) with target error~$\gamma=0$ is FPT for the parameter~$d+k$ (note that, for $\gamma > 0$, the problem is already W[1]-hard with respect to~$d$ for~$k=1$~\cite{FHN22}).

\begin{theorem}\label{thm:convexFPT}
  $\RELU{\mathcal{L}}^+$ can be solved in $2^{O(k^2d)}\poly(k,L)$ time for~$\gamma=0$, where~$L$ is the input bit-length.
\end{theorem}

Before giving the proof, we introduce some definitions.
For $I\subseteq[k]$, let~$R_I\subseteq\R^d$ be the \emph{active region} of the ReLUs in~$I$, that is,~$\mathbf{x}\in R_I$ if and only if
\begin{align*}
  \forall j\in I : \mathbf{w}_j\mathbf{x}+b_j&\ge 0,\\
  \forall j\in [k]\setminus I : \mathbf{w}_j\mathbf{x}+b_j&\le 0.
\end{align*}
Note that~$R_I$ could be empty.
Clearly, on each~$R_I$, $\phi$ is the affine function~$\sum_{j\in I}(\mathbf{w}_j\cdot\mathbf{x}+b_j)$.
Let $F_I\coloneqq \{(\mathbf{x},\phi(\mathbf{x}))\mid\mathbf{x}\in R_I)\}$ be the piece corresponding to~$I$.

The convexity of~$\phi$ now allows for a branching algorithm assigning the input data points to the at most~$2^k$ pieces. 

\begin{proof}[Proof of \Cref{thm:convexFPT}]
  Let~$(\mathbf{x}_1,y_1),\ldots,(\mathbf{x}_n,y_n)\in\R^{d+1}$,~$k\in\N$, and let~$L$ denote the overall number of input bits.
  The idea is to use a search tree algorithm to check whether the data can be exactly fitted with~$k$ (convex) ReLUs. To this end, we define~$2^k$ sets~$S_1,\ldots,S_{2^k}$ where each~$i\in[2^k]$ one-to-one corresponds to a certain subset~$I(i)\subseteq[k]$ of active ReLUs.
  For given point sets $S\subseteq\R^{d+1}$ and $S_i\subseteq\R^{d+1}$, $i\in[2^k]$, our algorithm checks whether the points in~$S$ can be exactly fitted by~$k$ ReLUs with the additional constraint
  that~$S_i\subseteq F_{I(i)}$ holds for each~$i\in[2^k]$. That is, the following (in)equalities must hold
  \begin{equation}\mathbf{x}\in R_{I(i)} \text{ and } \sum_{j\in I(i)}\mathbf{w}_j\mathbf{x} + b_j = y,\quad i\in[2^k], (\mathbf{x},y)\in S_i.\label{eqn:faceI}
  \end{equation}
\noindent
  \Cref{alg:exact-fit} depicts the pseudocode of our \texttt{ExactFit} algorithm.
  We solve an instance with an initial call where~$S\coloneqq \{(\mathbf{x}_1,y_1),\ldots,(\mathbf{x}_n,y_n)\}, S_1=S_2=\cdots=S_{2^k}\coloneqq \emptyset$.

  \LinesNumbered
  \begin{algorithm2e*}[t]
    \caption{\texttt{ExactFit}$(S,S_1,\ldots,S_{2^k})$}
    \label{alg:exact-fit}
    \DontPrintSemicolon
    \SetKwFunction{checkFeas}{check-feasibility}
    \SetKwFunction{fit}{ExactFit}
    \SetKwFunction{checkForced}{check-forced-points}
    \eIf{$S=\emptyset$}{
      \Return{$\checkFeas(S_1,\ldots,S_{2^k})$}\label{line:LP}
    }{
      choose $(\mathbf{x},y)\in S$\;
      $S\gets S\setminus \{(\mathbf{x},y)\}$\;
      \ForEach{$i=1,\ldots,2^k$}{\label{line:branch}
        $S_i\gets S_i\cup\{(\mathbf{x},y)\}$\;\label{line:put}
        $\checkForced(S,S_1,\ldots,S_{2^k})$\;\label{line:forced}
        $a \gets \fit(S, S_1,\ldots, S_{2^k})$\;\label{line:rec}
        \lIf{$a = $ \normalfont Yes}{\Return{\normalfont Yes}}
      }
    }
    \Return{\normalfont No}\;\label{line:No}
  \end{algorithm2e*}
  
  The correctness of~\Cref{alg:exact-fit} follows by induction on~$|S|$.
  For~$S=\emptyset$, we simply need to check whether the system~(\ref{eqn:faceI}) of linear (in)equalities is feasible. This can be done by solving a linear program with $k(d+1)$ variables and~$O(n)$ constraints in $O(\poly(k,L))$ time (this is done by \texttt{check-feasibility} in Line~\ref{line:LP}). 
  
  If~$S\neq\emptyset$ and~$(S,S_1,\ldots,S_{2^k})$ is a no-instance, then none of the recursive calls in Line~\ref{line:rec} will be successful (by induction). Hence, the algorithm correctly returns ``No'' in Line~\ref{line:No}.
  
  Now assume that~$(S,S_1,\ldots,S_{2^k})$ is a yes-instance.
  Then, any point~$(\mathbf{x},y)\in S$ must lie on some piece~$F_{I(i)}$. That is,~$(\mathbf{x},y)$ can be put into some~$S_i$.
  Hence, in Line~\ref{line:branch}, we branch into all~$2^k$ options.
  In each branch, we then check whether putting~$(\mathbf{x},y)$ into~$S_i$ also forces other points from~$S$ (due to assumed convexity) to be contained in some~$S_{i'}$ (this is done by \texttt{check-forced-points} in Line~\ref{line:forced}).
  We do this in order to achieve our claimed running time bound as we will show later.

  The pseudocode for this check is given in \Cref{alg:check-force}.
  \begin{algorithm2e*}[t]
  \caption{\texttt{check-forced-points}$(S,S_1,\ldots,S_{2^k})$}
  \label{alg:check-force}
  \DontPrintSemicolon
  \SetKwFunction{checkForced}{check-forced-points}
  \SetKwFunction{lb}{lower-bound}
    \ForEach{$(\mathbf{x},y)\in S$}{
      \ForEach{$i=1,\ldots,2^k$}{
        $\mu\gets \lb(\mathbf{x},i,S_1,\ldots,S_{2^k})$\;\label{line:lb}
        \If{$\mu = y$}{
          $S_i\gets S_i\cup\{(\mathbf{x},y)\}$\;
          $S\gets S\setminus \{(\mathbf{x},y)\}$\;
          \textbf{restart}\;\label{line:restart}
        }
        \If{$\mu > y$}{
          \textbf{reject branch} (\Return No)\;\label{line:reject}
        }
      }
    }
\end{algorithm2e*}
  The idea is to compute for each~$(\mathbf{x},y)\in S$ and each~$i\in[2^k]$ the lower bound
  \begin{align*}
    \mu \coloneqq \min_{\mathbf{w}_j,b_j} \sum_{j\in I(i)}(\mathbf{w}_j\mathbf{x} + b_j)
  \end{align*}
  subject to the constraints (\ref{eqn:faceI}), which again can be accomplished via linear programming in~$O(\poly(k,L))$ time. Note that both $\mu=+\infty$ (linear program is infeasible) and $\mu=-\infty$ (linear program is unbounded) are possible. This is done by \texttt{lower-bound} in Line~\ref{line:lb}.
  Now, note that~$\mu > y$ implies that
  \[\phi(\mathbf{x})=\sum_{j=1}^k[\mathbf{w}_j\mathbf{x}+b_j]_+ \ge \sum_{j\in I(i)}(\mathbf{w}_j\mathbf{x} + b_j) > y\]
  holds for every~$\phi$ satisfying~(\ref{eqn:faceI}).
  That is, we can reject~(Line~\ref{line:reject}) the current branch of \texttt{ExactFit}.
  If $\mu = y$, then we have
  \[\phi(\mathbf{x})\ge \sum_{j\in I(i)}(\mathbf{w}_j\mathbf{x} + b_j) = y\]
  for every~$\phi$ satisfying~(\ref{eqn:faceI}), and thus we can safely put~$(\mathbf{x},y)$ into~$S_i$. To see that this is correct, assume that a solution puts~$(\mathbf{x},y)\in F_{I'}$ for some~$I'\subseteq [k]$ with~$I'\neq I(i)$.
Then, we have
  \begin{align*}
    y &= \sum_{j\in I'}(\mathbf{w}_j\mathbf{x}+b_j)=\sum_{j\in I'\cap I(i)}(\mathbf{w}_j\mathbf{x}+b_j)+\sum_{j\in I'\setminus I(i)}(\mathbf{w}_j\mathbf{x}+b_j)\\
    &= \sum_{j\in I(i)}(\mathbf{w}_j\mathbf{x}+b_j)=\sum_{j\in I'\cap I(i)}(\mathbf{w}_j\mathbf{x}+b_j)+\sum_{j\in I(i)\setminus I'}(\mathbf{w}_j\mathbf{x}+b_j),
  \end{align*}
  which implies
  \[\sum_{j\in I'\setminus I(i)}(\mathbf{w}_j\mathbf{x}+b_j)=\sum_{j\in I(i)\setminus I'}(\mathbf{w}_j\mathbf{x}+b_j).\]
  Since~$\mathbf{x}\in R_{I'}$, it follows that the left sum is at least zero and the right sum is at most zero. Thus, both sums are zero and $\mathbf{w}_j\mathbf{x}+b_j=0$ holds for all $j\in(I'\setminus I(i))\cup(I(i)\setminus I')$, which shows that~$\mathbf{x}\in R_{I(i)}$.
  Thus, putting $(\mathbf{x},y)$ into~$S_i$ is correct.
  
  Note that adding a point to~$S_i$ adds new constraints to~(\ref{eqn:faceI}). Hence,
  we restart the procedure (Line~\ref{line:restart}) to check whether this forces new points.
  Overall, \texttt{check-forced-points} takes~$O(n^22^k\poly(k,L))\subseteq O(2^k\poly(k,L))$ time.

  As regards the correctness of~\Cref{alg:exact-fit} now, note that \texttt{check-forced-points} clearly never incorrectly rejects a branch of~\texttt{ExactFit} and never forces points incorrectly.
  Hence, one of the recursive calls in Line~\ref{line:rec} will correctly answer ``Yes'' (by induction), which proves the correctness.

  It remains to analyze the running time of~\Cref{alg:exact-fit}.
  Clearly, each call to the algorithm takes~$O(2^k\poly(k,L))$ time and recursively branches into~$2^k$ options.
  It remains to bound the depth of the recursion tree.
  To this end, note that the recursion stops as soon as~$S$ is empty or the current branch is rejected by~\Cref{alg:check-force}.
  We claim that the latter happens after at most~$k(d+1)+1$ recursive calls.

  To verify this claim, observe that the algorithm maintains the invariant that the linear program
  \[\min_{\mathbf{w}_j,b_j}\sum_{j\in I(i)}(\mathbf{w}_j\mathbf{x} +b_j)\quad \text{s.t. } (\ref{eqn:faceI})\]
  has a solution $\mu < y$ for every~$i\in[2^k]$ and $(\mathbf{x},y)\in S$. This invariant is achieved by checking for forced points in Line~\ref{line:forced}.
  Let $P\subseteq\R^{k(d+1)}$ be the polyhedron defined by~(\ref{eqn:faceI}) in the variables~$(\mathbf{w}_j,b_j)_{j=1,\ldots,k}$.
  Now, adding a point~$(\mathbf{x},y)$ to some~$S_i$ (Line~\ref{line:put}) adds constraints to~(\ref{eqn:faceI}) which yield a polyhedron contained in
  \[P'\coloneqq P\cap \{(\mathbf{w}_j,b_j)_{j=1,\ldots,k}\mid \sum_{j\in I(i)}(\mathbf{w}_j\mathbf{x}+b_j)=y\}.\]
  By the above invariant, there exists a $(\mathbf{w}_j,b_j)_{j=1,\ldots,k}\in P$ with $\sum_{j\in I(i)}(\mathbf{w}_j\mathbf{x}+b_j)<y$. Hence, $\aff(P')\subsetneq \aff(P)$ and~$\dim(P')<\dim(P)$. That is, each recursive call decreases the dimension of the feasible polyhedron. 
  Thus, after at most~$k(d+1)+1$ recursive calls we reach an empty polyhedron, in which case the current branch is rejected.
  
  To sum up, we obtain an overall running time of~$2^{O(k^2d)}\poly(k,L)$.
\end{proof}

Clearly, \Cref{thm:convexFPT} analogously holds for the concave case where~$a_j=-1$ for all~$j\in[k]$.
Note that if positive and negative coefficients~$a_j$ are allowed, then our search tree approach of \Cref{alg:exact-fit} does not work since we cannot check for forced points anymore which is necessary to ensure a bounded recursion depth.
Indeed, \Cref{thm:k=4} implies that this approach cannot work already for~$k=4$.
It is unclear whether this issue can be resolved for~$k=2$ or~$k=3$.

\section{Conclusion}

We closed several gaps in the literature regarding the computational complexity of training two-layer neural networks. Our results give some insight into the geometry of functions realized by such networks and yield a better understanding of their complexity and expressiveness.
We thereby settled the border of computational tractability almost completely.
The remaining open questions are the following:

\begin{itemize}
  \item Is the problem with~$d=2$ in FPT when parameterized by~$k$? This is open for both ReLUs and linear thresholds.
    Note that W[1]-hardness with respect to~$k$ for any constant~$d$ would answer Question~2 negatively.
  \item Is the case $\gamma=0$ and~$k\in\{2,3\}$ in FPT with parameter~$d$ for ReLUs?
\end{itemize}

In a broader context, open directions are to further study the computational complexity in appropriate approximate settings, draw further conclusions on generalization, and understand deeper network architectures.

\bibliographystyle{abbrvnat}
\bibliography{ref}

\end{document}